\documentclass[aps,pra,final,twocolumn,letterpaper,superscriptaddress,preprintnumbers,accepted=2019-09-02]{Quantum_Template/quantumarticle}
\pdfoutput=1
\usepackage{graphicx}                            
\usepackage{epstopdf}
\usepackage{amsmath} 
\usepackage{amssymb}
\usepackage{float}
\usepackage{amsthm}
\usepackage{enumerate}
\usepackage{qcircuit}
\usepackage{hyperref}
\usepackage{complexity}
\usepackage[usenames,dvipsnames]{color}
\usepackage[normalem]{ulem}
\usepackage{tikz}
\usetikzlibrary{arrows}
\usepackage[all]{nowidow}
\usepackage[numbers,sort&compress]{natbib}

\usepackage[utf8]{inputenc}
\usepackage[english]{babel}
\usepackage[T1]{fontenc}
\usepackage{hyperref}

\usepackage{xcolor}
\newcommand\hl[1]{%
  \bgroup
  \hskip0pt\color{blue!80!black}%
  #1%
  \egroup
}

\newcommand{\bra}[1]{\left\langle #1 \right|}
\newcommand{\ket}[1]{\left|#1\right\rangle}
\newcommand{\kettilde}[1]{|\tilde{#1}\rangle}
\newcommand{\bratilde}[1]{\langle\tilde{#1}|}
\newcommand{\braket}[2]{\left\langle#1 |  #2\right\rangle}

\newcommand{\Tr}{\text{Tr}}
\newcommand{\Cor}{\text{Cor}}
\newcommand{\dist}{\text{dist}}
\newcommand{\rank}{\text{rank}}

\newcommand{\be}{\begin{equation}}
\newcommand{\ee}{\end{equation}}

\newtheorem{definition}{Definition}

\newtheorem{theorem}{Theorem}
\newtheorem{corollary}{Corollary}
\newtheorem{lemma}{Lemma}

\newtheorem{problem}{Problem}
\newtheorem{claim}{Claim}

\def\nn{\nonumber\\}

\newcommand{\nc}{\newcommand}
\nc{\rnc}{\renewcommand}

\nc\bbC{\mathbb{C}}

\nc\bbF{\mathbb{F}}
\nc\bbM{\mathbb{M}}
\nc\bbN{\mathbb{N}}
\nc\bbR{\mathbb{R}}
\nc\bbZ{\mathbb{Z}}

\begin{document}

\title{Locally accurate MPS approximations for ground states of one-dimensional gapped local Hamiltonians}

\author{Alexander M. Dalzell}
\affiliation{Institute for Quantum Information and Matter, California Institute of
Technology, Pasadena, CA 91125, USA
}
\author{Fernando G.S.L. Brand\~ao}
\affiliation{Institute for Quantum Information and Matter, California Institute of
Technology, Pasadena, CA 91125, USA
}
\affiliation{Google LLC, Venice, CA 90291, USA}

\begin{abstract}

A key feature of ground states of gapped local 1D Hamiltonians is their relatively low entanglement --- they are well approximated by matrix product states (MPS) with bond dimension scaling polynomially in the length $N$ of the chain, while general states require a bond dimension scaling exponentially. We show that the bond dimension of these MPS approximations can be improved to a constant, independent of the chain length, if we relax our notion of approximation to be more local: for all length-$k$ segments of the chain, the reduced density matrices of our approximations are $\epsilon$-close to those of the exact state. If the state is a ground state of a gapped local Hamiltonian, the bond dimension of the approximation scales like $(k/\epsilon)^{1+o(1)}$, and at the expense of worse but still $\text{poly}(k,1/\epsilon)$ scaling of the bond dimension, we give an alternate construction with the additional features that it can be generated by a constant-depth quantum circuit with nearest-neighbor gates, and that it applies generally for any state with exponentially decaying correlations. For a completely general state, we give an approximation with bond dimension $\exp(O(k/\epsilon))$, which is exponentially worse, but still independent of $N$. Then, we consider the prospect of designing an algorithm to find a local approximation for ground states of gapped local 1D Hamiltonians. When the Hamiltonian is translationally invariant, we show that the ability to find $O(1)$-accurate local approximations to the ground state in $T(N)$ time implies the ability to estimate the ground state energy to $O(1)$ precision in $O(T(N)\log(N))$ time.

\end{abstract}
\maketitle

\section{Introduction}

In nature, interactions between particles act locally, motivating the study of many-body Hamiltonians consisting only of terms involving particles spatially near each other. An important method that has emerged from this course of study is the Density Matrix Renormalization Group (DMRG) algorithm \cite{white1992density,white1993density}, which aims to find a description of the ground state of local Hamiltonians on a one-dimensional chain of sites. DMRG has been an indispensable tool for research in many-body physics, but the rationale for its empirical success did not become fully apparent until long after it was being widely used. Its eventual justification required two ingredients: first, that the method can be recast \cite{klumper1993matrix,weichselbaum2009variational} as a variational algorithm that minimizes the energy over the set of matrix product states (MPS), a tensor network ansatz for 1D systems; and second, that the MPS ansatz set actually contains a good approximation to the ground state, at least whenever the Hamiltonian has a nonzero spectral gap. More specifically, Hastings \cite{hastings2007area} showed that ground states of gapped local Hamiltonians on chains with $N$ sites can be approximated to within trace distance $\epsilon$ by an MPS with bond dimension (a measure of the complexity of the MPS) only $\text{poly}(N,1/\epsilon)$, exponentially smaller than what is needed to describe an arbitrary state on the chain. Even taking into account these observations, DMRG is a heuristic algorithm and is not guaranteed to converge to the global energy minimum as opposed to a local minimum; however, a recently developed alternative algorithm \cite{landau2015polynomial, arad2017rigorous,roberts2017implementation}, sometimes referred to as the ``Rigorous RG'' (RRG) algorithm, avoids this issue and provides a way one can guarantee finding an $\epsilon$-approximation to the ground state in $\text{poly}(N,1/\epsilon)$ time.

These are extremely powerful results, but their value breaks down when the chain becomes very long. The bond dimension required to describe the ground state grows relatively slowly, but it still diverges with $N$. Meanwhile, if we run the RRG algorithm on longer and longer chains, we will eventually encounter an $N$ too large to handle given finite computational resources. Indeed, often we wish to learn something about the ground state in the thermodynamic limit ($N\rightarrow \infty$) but in this case these results no longer apply. Analogues of DMRG for the thermodynamic limit \cite{mcculloch2008infinite,vidal2007classical,haegeman2011time,zauner2018variational,ostlund1995thermodynamic,vanderstraeten2019tangent} --- methods that, for example, optimize over the set of constant bond dimension ``uniform MPS'' consisting of the same tensor repeated across the entire infinite chain --- have been implemented with successful results, but these methods lack the second ingredient that justified DMRG: it is not clear how large we must make the bond dimension to guarantee that the set over which we are optimizing contains a good approximation to the ground state.

Progress toward this ingredient can be found work by Huang \cite{huang2015computing} (and later by Schuch and Verstraete \cite{schuch2017matrix}), who showed that the ground state of a gapped local 1D Hamiltonian can be approximated \textit{locally} by a matrix product operator (MPO) --- a 1D tensor network object that corresponds to a  (possibly mixed) density operator as opposed to a quantum state vector --- with bond dimension independent of $N$ and polynomial in the inverse local approximation error. Here their approximation sacrifices global fidelity with the ground state, which decays exponentially with the chain length, in exchange for \textit{constant} bond dimension, while retaining high fidelity with the ground state reduced density matrices on all segments of the chain with constant length. In other words, the statistics for measurements of local operators are faithfully captured by the MPO approximation, a notion of approximation that is often sufficient in practice since many relevant observables, such as the individual terms in the Hamiltonian, are local. 

However, the result does not provide the necessary ingredient to justify infinite analogues of DMRG because MPO do not make a good ansatz class for variational optimization algorithms. One can specify the matrix elements for an MPO, but the resulting operator will only correspond to a valid quantum state if it is positive semi-definite, and verifying that this is the case is difficult: it is $\mathsf{NP}$-hard for finite chains, and in the limit $N \rightarrow \infty$ it becomes undecidable \cite{kliesch2014matrix}. Thus, if we attempt to perform variational optimization over the set of constant bond dimension MPO, we can be certain that our search space contains a good local approximation to the ground state, but we have no way of restricting our search only to the set of valid quantum states; ultimately the minimal energy MPO we find may not correspond to any quantum state at all.  

In this work, we fix this problem by showing an analogous result for MPS instead of MPO. We show that for any gapped nearest-neighbor Hamiltonian on a 1D chain with $N$ sites, and for any parameters $k$ and $\epsilon$, there is an MPS representation of a state $\kettilde{\psi}$ with bond dimension $\text{poly}(k, 1/\epsilon)$ such that the reduced density matrix of $\kettilde{\psi}\bratilde{\psi}$ on any contiguous region of length $k$ is $\epsilon$-close in trace distance to that of the true ground state.  Importantly, the bond dimension is independent of $N$. For general states (including ground states of non-gapped local Hamiltonians), we give a construction with bond dimension that is also independent of $N$ but exponential in $k/\epsilon$. This extends a previous result \cite{fannes1992abundance} that formally implied the existence of a uniform MPS approximation of this type when the state is translationally invariant and $N=\infty$, albeit without explicit attention paid to the dependence of the bond dimension on the locality $k$ or approximation error $\epsilon$, or to an improvement therein when the state is the ground state of a gapped Hamiltonian. Thus, we provide the missing ingredient for variational algorithms in the thermodynamic limit as we show that a variational set over a MPS with bond dimension independent in $N$ and polynomial in $1/\epsilon$ contains a state that simultaneously captures all the local properties of the ground state.

We present two proofs of our claim about ground states of gapped Hamiltonians. The first yields superior scaling of the bond dimension, which grows asymptotically slower than $(k/\epsilon)^{1+\delta}$ for any $\delta > 0$; however, it constructs an MPS approximation $\kettilde{\psi}$ that is long-range correlated and non-injective. In contrast, the second proof constructs an approximation that is injective and can be generated by a constant-depth quantum circuit with nearest-neighbor gates, while retaining $\text{poly}(k,1/\epsilon)$ bond dimension. The latter construction also follows merely from the assumption that the state has exponential decay of correlations. The proof idea originates with a strategy first presented in \cite{verstraete2006matrix} and constructs $\kettilde{\psi}$ by beginning with the true ground state $\ket{\psi}$ and applying three rounds of operations: first, a set of unitaries that, intuitively speaking, removes the short-range entanglement from the chain; second, a sequence of rank 1 projectors that zeroes-out the long-range entanglement; and third, the set of inverse unitaries from step 1 to add back the short-range entanglement. Intuitively, the method works because ground states of gapped Hamiltonians have a small amount of long-range entanglement. The non-trivial part is arguing that the local properties are preserved even as the small errors induced in step 2 accumulate to bring the global fidelity with the ground state to zero. The fact that $\kettilde{\psi}$ can be produced by a constant-depth quantum circuit acting on an initial product state suggests the possibility of an alternative variational optimization algorithm using the set of constant-depth circuits (a strict subset of constant-bond-dimension MPS) as the variational ansatz. Additionally, we note that the disentangle-project-reentangle process that we utilize in our proof might be of independent interest as a method for truncating the bond dimension of MPS. We can bound the truncation error of this method when the state has exponentially decaying correlations.

We also consider the question of whether these locally approximate MPS approximations can be rigorously found (\`a la RRG) more quickly than their globally approximate counterparts (and if they can be found at all in the thermodynamic limit). We prove a reduction for ground states of translationally invariant Hamiltonians showing that finding approximations to local properties to even a fixed $O(1)$ precision implies being able to find an approximation to the ground state energy to $O(1)$ precision with only $O(\log(N))$ overhead. Since strategies for estimating the ground state energy typically involve constructing a globally accurate approximation to the ground state, this observation gives us an intuition that it may not be possible to find the local approximation much more quickly than the global approximation, despite the fact that the bond dimensions required for the two approximations are drastically different.

\section{Background} \label{sec:background}

\subsection{One-dimensional local Hamiltonians}

In this paper, we work exclusively with gapped nearest-neighbor 1D Hamiltonians that have a unique ground state. Our physical system is a set of $N$ sites, arranged in one dimension on a line with open boundary conditions (OBC), each with its own Hilbert space $\mathcal{H}_i$ of dimension $d$. The Hamiltonian $H$ consists of terms $H_{i,i+1}$ that act non-trivially only on $\mathcal{H}_i$ and $\mathcal{H}_{i+1}$:
\begin{equation}
    H = \sum_{i=1}^{N-1} H_{i,i+1}.
\end{equation}
We will always require that $H_{i,i+1}$ be positive semi-definite and satisfy $\lVert H_{i,i+1} \rVert \leq 1$ for all $i$, where $\lVert \cdot \rVert$ is the operator norm. When this is not the case it is always possible to rescale $H$ so that it is. We call $H$ translationally invariant if $H_{i,i+1}$ is the same for all $i$. We will also always assume that $H$ has a unique ground state $\ket{\psi}$ with energy $E$ and an energy gap $\Delta > 0$ to its first excited state.  We let $\rho = \ket \psi \bra \psi$ refer to the (pure) density matrix representation of the ground state. For any density matrix $\sigma$ and subregion $X$ of the chain, we let $\sigma_X$ refer to $\Tr_{X^c}(\sigma)$, the reduced density matrix of $\sigma$ after tracing out the complement $X^c$ of $X$. 

Theorems \ref{thm:improvedbd} and \ref{thm:mainthm} will make statements about efficiently approximating the ground state of such Hamiltonians with matrix product states, and Theorem \ref{thm:reduction} is a statement about algorithms that estimate the ground state energy $E$ or approximate the expectation $\bra{\psi} O \ket{\psi}$ of a local observable $O$ in the ground state. 

\subsection{Matrix product states and matrix product operators}

It is often convenient to describe states with one-dimensional structure using the language of matrix product states (MPS). 
\begin{definition}[Matrix product state]
A matrix product state (MPS) $\ket \eta$ on $N$ sites of local dimension $d$ is specified by $Nd$ matrices $A_j^{(i)}$ with $i = 1, \ldots, d$ and $j = 1, \ldots, N$. The matrices $A_1^{(i)}$ are $1 \times \chi$ matrices and $A_N^{(i)}$ are $\chi \times 1$ matrices, with the rest being $\chi \times \chi$. The state is defined as
\begin{equation}
\ket \eta = \sum_{i_1 = 1}^d \ldots \sum_{i_N = 1}^d A_1^{(i_1)}\ldots A_N^{(i_N)} \ket{i_1 \ldots i_N}.
\end{equation}
The parameter $\chi$ is called the \textit{bond dimension} of the MPS.
\end{definition}

The same physical state has many different MPS representations, although one may impose a canonical form \cite{perez2006matrix} to make the representation unique. The bond dimension of the MPS is a measure of the maximum amount of entanglement across any ``cut'' dividing the state into two contiguous parts. More precisely, if we perform a Schmidt decomposition on a state $\ket\eta$ across every possible cut, the maximum number of non-zero Schmidt coefficients (i.e.~Schmidt rank) across any of the cuts is equal to the minimum bond dimension we would need to exactly represent $\ket\eta$ as an MPS \cite{vidal2003efficient}. Thus to show a state has an MPS representation with a certain bond dimension, it suffices to bound the Schdmidt rank across all the cuts. This line of reasoning shows that a product state, which has no entanglement, can be written as an MPS with bond dimension 1. Meanwhile, a general state with any amount of entanglement can always be written as an MPS with bond dimension $d^{N/2}$. 

A cousin of matrix product states are matrix product operators (MPO).

\begin{definition}[Matrix product operator]

A matrix product operator (MPO) $\sigma$ on $N$ sites of local dimension $d$ is specified by $Nd^2$ matrices $A_j^{(i)}$ with $i = 1, \ldots, d^2$ and $j = 1, \ldots, N$. The matrices $A_1^{(i)}$ are $1 \times \chi$ matrices and $A_N^{(i)}$ are $\chi \times 1$ matrices, with the rest being $\chi \times \chi$. The operator is defined as
\begin{equation}
\sigma = \sum_{i_1 = 1}^{d^2} \ldots \sum_{i_N = 1}^{d^2} A_1^{(i_1)}\ldots A_N^{(i_N)} \sigma_{i_1} \otimes \ldots \otimes \sigma_{i_N},
\end{equation}
where $\{\sigma_{i}\}_{i=1}^{d^2}$ is a basis for operators on a single site.
The parameter $\chi$ is called the \textit{bond dimension} of the MPO.
\end{definition}

However, MPO representations have the issue that specifying a set of matrices $A_j^{(i)}$ does not always lead to an operator $\sigma$ that is positive semi-definite, which is a requirement for the MPO to correspond to a valid quantum state. Checking positivity of an MPO in general is $\mathsf{NP}$-hard for chains of length $N$ and undecidable for infinite chains \cite{kliesch2014matrix}.

\subsection{Notions of approximation}

We are interested in the existence of an MPS that approximates the ground state $\ket \psi$. We will have both a global and a local notion of approximation, which we define here. We will employ two different distance measures at different points in our theorems and proofs, the purified distance \cite{gilchrist2005distance, tomamichel2009fully} and the trace distance.

\begin{definition}[Purified distance]
If $\sigma$ and $\sigma'$ are two normalized states on the same system, then
\begin{equation}
D(\sigma, \sigma') = \sqrt{1-\mathcal{F}(\sigma,\sigma')^2}
\end{equation}
is the purified distance between $\sigma$ and $\sigma'$, where $\mathcal{F}(\sigma,\sigma') = \Tr(\sqrt{\sigma^{1/2}\sigma'\sigma^{1/2}})$ denotes the fidelity between $\sigma$ and $\sigma'$.
\end{definition}

\begin{definition}[Trace distance]
If $\sigma$ and $\sigma'$ are two normalized states on the same system, then
\begin{equation}
D_1(\sigma, \sigma') = \frac{1}{2}\lVert \sigma - \sigma' \rVert_1 = \frac{1}{2}\Tr(\lvert \sigma-\sigma' \rvert)
\end{equation}
is the trace distance between $\sigma$ and $\sigma'$.
\end{definition}

\begin{lemma}[\cite{tomamichel2009fully}]\label{lem:purifiedvstrace}
\begin{equation}
D_1(\sigma,\sigma') \leq D(\sigma,\sigma') \leq \sqrt{2D_1(\sigma,\sigma')}.
\end{equation}
\end{lemma}
We also note that $D_1(\sigma,\sigma') = D(\sigma,\sigma')$ if $\sigma$ and $\sigma'$ are both pure. If the trace distance between $\rho$ and $\sigma$ is small then we would say $\sigma$ is a good global approximation to $\rho$. We are also interested in a notion of distance that is more local. 

\begin{definition}[$k$-local purified distance]
If $\sigma$ and $\sigma'$ are two normalized states on the same system, then the $k$-local purified distance between $\sigma$ and $\sigma'$ is
\begin{equation}
D^{(k)}(\sigma, \sigma') = \max_{X: \lvert X \rvert = k}D(\sigma_X,\sigma'_X),
\end{equation}
where the max is taken over all contiguous regions $X$ consisting of $k$ sites.
\end{definition}

\begin{definition}[$k$-local trace distance]
If $\sigma$ and $\sigma'$ are two normalized states on the same system, then the $k$-local trace distance between $\sigma$ and $\sigma'$ is
\begin{equation}
D_1^{(k)}(\sigma, \sigma') := \max_{X: \lvert X \rvert = k}D_1(\sigma_X,\sigma'_X),
\end{equation}
where the max is taken over all contiguous regions $X$ consisting of $k$ sites.
\end{definition}

Note that these quantities lack the property that $0=D^{(k)}(\sigma,\sigma')=D_1^{(k)}(\sigma,\sigma')$ implies $\sigma = \sigma'$,\footnote{To see this consider the simple counterexample where $k=2$, $\sigma = \ket{\eta}\bra{\eta}$, $\sigma' = \ket{\nu}\bra{\nu}$, with $\ket{\eta}= (\ket{000}+\ket{111})/\sqrt{2}$, $\ket{\nu} = (\ket{000}-\ket{111})/\sqrt{2}$. In fact here $\braket{\nu}{\eta}=0$. This counterexample can be generalized to apply for any $k$.} but they do satisfy the triangle inequality. It is also clear that taking $k=N$ recovers our notion of global distance: $D^{(N)}(\sigma,\sigma') = D(\sigma,\sigma')$ and $D_1^{(N)}(\sigma,\sigma') = D_1(\sigma,\sigma')$.

\begin{definition}[Local approximation]
We say a state $\sigma$ on a chain of $N$ sites is a $(k,\epsilon)$-local approximation to another state $\sigma'$ if $D_1^{(k)}\left(\sigma, \sigma'\right) \leq \epsilon$.
\end{definition}

As we discuss in the next subsection, previous results show that $\ket\psi$ has a good global approximation $\kettilde{\psi}$ that is an MPS with bond dimension that scales like a polynomial in $N$. We will be interested in the question of what bond dimension is required when what we seek is merely a good local approximation. 

\subsection{Previous results}

\subsubsection{Exponential decay of correlations and area laws}

A key fact shown by Hastings \cite{hastings2004lieb} (see also \cite{hastings2006spectral,hastings2004locality, nachtergaele2006lieb} for improvements and extensions) about nearest-neighbor 1D Hamiltonians with a non-zero energy gap is that the ground state $\ket{\psi}$ has exponential decay of correlations.

\begin{definition}[Exponential decay of correlations]
A pure state $\sigma = \ket \eta \bra \eta$ on a chain of sites is said to have $(t_0,\xi)$-exponential decay of correlations if for every $t \geq t_0$ and every pair of regions $A$ and $C$ separated by at least $t$ sites
\begin{align}
& \Cor(A:C)_{\ket{\eta}} \nonumber\\
& := \max_{\lVert M \rVert, \lVert N \rVert \leq 1} \Tr\left((M \otimes N) (\sigma_{AC}-\sigma_A \otimes \sigma_C)\right) \nonumber\\
& \leq \exp(-t/\xi).
\end{align}
The smallest $\xi$ for which $\sigma$ has $(t_0,\xi)$-exponential decay of correlations for some $t_0$ is called the correlation length of $\sigma$.
\end{definition}

\begin{lemma}
If $\ket \psi$ is the unique ground state of a Hamiltonian $H = \sum_i H_{i,i+1}$ with spectral gap $\Delta$, then $\ket\psi$ has $(t_0,\xi)$-exponential decay of correlations for some $t_0 = O(1)$ and $\xi = O(1/\Delta)$.
\end{lemma}
\begin{proof}
This statement is implied by Theorem 4.1 of \cite{nachtergaele2006lieb}.
\end{proof}

While the exponential decay of correlations holds for lattice models in any spatial dimension, the other results we discuss are only known to hold in one dimension.

For example, in one dimension it has been shown that ground states of gapped Hamiltonians obey an area law, that is, the entanglement entropy of any contiguous region is bounded by a constant times the length of the boundary of that region, which in one dimension is just a constant. This statement was also first proven by Hastings in \cite{hastings2007area} where it was shown that for any contiguous region $X$
\begin{equation}
    S(\rho_X) \leq \exp(O(\log(d)/\Delta)),
\end{equation}
which is independent of the number of sites in $X$, where $S$ denotes the von Neumann entropy $S(\sigma) = -\Tr(\sigma \log \sigma)$. The area law has since been improved \cite{arad2013area, arad2017rigorous} to 
\begin{equation}
    S(\rho_X) \leq \tilde{O}(\log^3(d)/\Delta),
\end{equation}
where the $\tilde{O}$ signifies a suppression of factors that scale logarithmically with the quantity stated. 

It was also discovered that an area law follows merely from the assumption of exponential decay of correlations in one dimension: if a pure state $\rho$ has $(t_0,\xi)$-exponential decay of correlations, then it satisfies \cite{brandao2015exponential}
\begin{equation}
    S(\rho_X) \leq t_0\exp(\tilde{O}(\xi \log(d))).
\end{equation}

\subsubsection{Efficient global MPS approximations}
The area law is closely related to the existence of an efficient MPS approximation to the ground state. To make this implication concrete, one needs an area law using the $\alpha$-Renyi entropy for some value of $\alpha$ with $0 < \alpha < 1$ \cite{verstraete2006matrix}, where the Renyi entropy is given by $S_\alpha(\rho_X) = -\Tr(\log(\rho_X^\alpha))/(1-\alpha)$. An area law for the von Neumann entropy (corresponding to $\alpha=1$) is not alone sufficient \cite{schuch2008entropy}. However, for all of the area laws mentioned above, the techniques used are strong enough to also imply the existence of efficient MPS approximations, and, moreover, area laws have indeed been shown for the $\alpha$-Renyi entropy \cite{huang2014area} with $0 < \alpha < 1$.

Hastings' \cite{hastings2007area} original area law implied the existence of a global $\epsilon$-approximation $\kettilde{\psi}$ for $\ket{\psi}$ with bond dimension
\begin{equation}
\chi = e^{\tilde{O}\left(\frac{\log(d)}{\Delta}\right)} \left(\frac{N}{\epsilon}\right)^{O\left(\frac{\log(d)}{\Delta}\right)}.
\end{equation}

The improved area law in \cite{arad2013area, arad2017rigorous} yields a better scaling for $\chi$ which is asymptotically sublinear in $N$:
\begin{equation}
\chi = e^{\tilde{O}\left(\frac{\log^3(d)}{\Delta}\right)} \left(\frac{N}{\epsilon}\right)^{\tilde{O}\left(\frac{\log(d)}{(\Delta\log(N/\epsilon))^{1/4}}\right)}.
\end{equation}

Finally, the result implied only from exponential decay of correlations \cite{brandao2015exponential} is
\begin{equation}
\chi = e^{t_0e^{\tilde{O}\left(\xi \log(d)\right)}} \left(\frac{N}{\epsilon}\right)^{\tilde{O}\left(\xi\log(d)\right)}.
\end{equation}

Crucially, if the local Hilbert space dimension $d$ and the gap $\Delta$ (or alternatively, the correlation length $\xi$) are taken to be constant, then all three results read $\chi = \text{poly}(N, 1/\epsilon)$.

\subsubsection{Existence of MPS approximations in the thermodynamic limit}

The aforementioned results, which describe explicit bounds on the bond dimension needed for good MPS approximations, improved upon important prior work that characterized which states can be approximated by MPS in the first place. Of course, any state on a finite chain can be exactly described by an MPS, but the question of states living on the infinite chain is more subtle. In \cite{fannes1992finitely}, the proper mathematical framework was developed to study MPS, which they call \textit{finitely correlated states}, in the limit of infinite system size, and in \cite{fannes1992abundance} it was shown that any translationally invariant state on the infinite chain can be approximated arbitrarily well by a uniform (translationally invariant) MPS in the following sense: for any translationally invariant pure state $\rho$ there exists a \textit{net} --- a generalization of a sequence --- of translationally invariant MPS $\rho_\alpha$ for which the expectation value $\Tr(\rho_\alpha A)$ of any finitely supported observable $A$ converges to $\Tr(\rho A)$. An implication of this is that if we restrict to observables $A$ with support on only $k$ contiguous sites, there exists a translationally invariant MPS that approximates the expectation value of all $A$ to arbitrarily small error $\epsilon$. Thus, they established that local approximations for translationally invariant states exist within the set of translationally invariant MPS, but provided no discussion of the bond dimension required for such an approximation, and did not explicitly consider the case where the state is the ground state of a gapped, local Hamiltonian. 

Our Theorem \ref{thm:improvedbd}, which is stated in the following section, may be viewed as a generalization and improvement on this work in several senses. Most importantly, we present a construction for which a bound on the bond dimension can be explicitly obtained. This bound scales like $\text{poly}(k,1/\epsilon)$ when the state is a ground state of a gapped nearest-neighbor Hamiltonian, and exponentially in $k/\epsilon$ when it is a general state.  Furthermore, our method works for states on the finite chain that are not translationally invariant, where it becomes unclear how the methods of this previous work would generalize.

\subsubsection{Constant-bond-dimension MPO local approximations}

The problem of finding matrix product operator representations that capture all the local properties of a state has been studied before. Huang \cite{huang2015computing} showed the existence of a positive semi-definite MPO $\rho^\chi$ with bond dimension
\begin{equation}\label{eq:HuangMPO}
    \chi = e^{\tilde{O}\left(\frac{\log^3(d)}{\Delta}+\frac{\log(d)\log^{3/4}(1/\epsilon)}{\Delta^{1/4}}\right)}  = (1/\epsilon)^{o(1)}
\end{equation}
that is a $(2,\epsilon)$-local approximation to the true ground state $\rho$, where $o(1)$ indicates the quantity approaches 0 as $1/\epsilon \rightarrow \infty$. Crucially, this is independent of the length of the chain $N$. Additionally, because the Hamiltonian is nearest-neighbor, we have $\Tr(H\rho^\chi)-\Tr(H\rho) \leq (N-1)\epsilon$, i.e., the energy per site (energy density) of the state $\rho^\chi$ is within $\epsilon$ of the ground state energy density. Huang constructs this MPO explicitly and notes it is a convex combination over pure states which themselves are MPS with bond dimension independent of $N$. Thus, one of these MPS must have energy density within $\epsilon$ of the 
ground state energy density. However, it is not guaranteed (nor is it likely) that one of these constant-bond-dimension MPS is also a good local approximation to the ground state; thus our result may be viewed as an improvement on this front as we show the existence not only of a low-energy-density constant-bond-dimension MPS, but also one that is a good local approximation to the ground state.

An alternative MPO construction achieving the same task was later given in \cite{schuch2017matrix}. In this case, the MPO is a $(k,\epsilon)$-local approximation to the ground state and has bond dimension 
\begin{equation}
    \chi=(k/\epsilon)e^{\tilde{O}\left(\frac{\log^3(d)}{\Delta}+\frac{\log(d)\log^{3/4}(k/\epsilon^3)}{\Delta^{1/4}}\right)} = (k/\epsilon)^{1+o(1)}.
\end{equation}

The idea they use is simple. They break the chain into blocks of size $l$, which is much larger than $k$. On each block they construct a constant bond dimension MPO that closely approximates the reduced density matrix of the ground state on that block, which is easy since each block has constant length and they must make only a constant number of bond truncations to the exact state. The tensor product of these MPO will be an MPO on the whole chain that is a good approximation on any length-$k$ region that falls within one of the larger length $l$ blocks, but not on a region that crosses the boundary between blocks. To remedy this, they take the mixture of MPO formed by considering all $l$ translations of the boundaries between the blocks. Now as long as $l$ is much larger than $k$, any region of length $k$ will only span the boundary between blocks for a small fraction of the MPO that make up this mixture, and the MPO will be a good local approximation.

This same idea underlies our proof of Theorem \ref{thm:improvedbd}, with the complication that we seek a pure state approximation and cannot take a mixture of several MPS. Instead, we combine the translated MPS in superposition, which brings new but manageable challenges.

\section{Statement of results} \label{sec:results}

\subsection{Existence of local approximation}

\begin{theorem}\label{thm:improvedbd}
Let $\ket \psi$ be a state on a chain of $N$ sites of local dimension $d$. For any $k$ and $\epsilon$ there exists an MPS $\kettilde{\psi}$ with bond dimension at most $\chi$ such that
\begin{enumerate} [(1)]
\item $\kettilde{\psi}$ is a $(k,\epsilon)$-local approximation to $\ket{\psi}$
\item $
    \chi = e^{O(k\log(d)/\epsilon)}
$
\end{enumerate}
provided that $N$ is larger than some constant $N_0 = O(k^3/\epsilon^3)$ that is independent of $\ket{\psi}$.

If $\ket{\psi}$ has $(t_0,\xi)$-exponential decay of correlations, then the bound on the bond dimension can be improved to
\begin{enumerate} [(2')]
\item $
    \chi = e^{t_0e^{\tilde{O}\left(\xi\log(d)\right)}} (k/\epsilon^3)^{O\left(\xi\log(d)\right)}
$
\end{enumerate}
with $N_0 = O(k^2/\epsilon^2) + t_0\exp(\tilde{O}(\xi \log(d))$
and if, additionally, $\ket{\psi}$ is the unique ground state of a nearest-neighbor 1D Hamiltonian $H$ with spectral gap $\Delta$, it can be further improved to
\begin{enumerate}[(2'')]
    \item $\chi =  (k/\epsilon)e^{\tilde{O}\left(\frac{ \log^{3}(d)}{\Delta}+\frac{\log(d)}{\Delta^{1/4}}\log^{3/4}(k/\epsilon^3)\right)}$
\end{enumerate}
with $N_0 = O(k^2/\epsilon^2) + \tilde{O}(\log(d)/\Delta^{3/4})$. Here $\chi$ is asymptotically equivalent to $(k/\epsilon)^{1+o(1)}$ where $o(1)$ indicates that the quantity approaches 0 as $(k/\epsilon) \rightarrow \infty$.
\end{theorem}

However, the state $\kettilde{\psi}$ that we construct in the proof of Theorem \ref{thm:improvedbd} is long-range correlated and cannot be generated from a constant-depth quantum circuit. Thus, while $\kettilde{\psi}$ is a good local approximation to the ground state $\ket{\psi}$ of $H$, it is not the exact ground state of any gapped local 1D Hamiltonian. 

Next, we show that it remains possible to approximate the state even when we require the approximation to be produced by a constant-depth quantum circuit; the scaling of the bond dimension is faster in $k$ and $1/\epsilon$, but it is still polynomial. 

\begin{theorem}\label{thm:mainthm}
Let $\ket \psi$ be a state on a chain of $N$ sites of local dimension $d$. If $\ket{\psi}$ has $(t_0,\xi)$-exponential decay of correlations, then, for any $k$ and $\epsilon$, there is an MPS $\kettilde{\psi}$ with bond dimension at most $\chi$ such that 
\begin{enumerate}[(1)]
    \item $\kettilde{\psi}$ is a $(k,\epsilon)$-local approximation to $\ket{\psi}$
    \item $
    \chi = e^{t_0e^{\tilde{O}\left(\xi\log(d)\right)}} \left(k/\epsilon^2\right)^{O\left(\xi^2\log^2(d)\right)}$
\item $\kettilde{\psi}$ can be prepared from the state $\ket{0}^{\otimes N}$ by a quantum circuit that has depth $\tilde{O}(\chi^2)$ and consists only of unitary gates acting on neighboring pairs of qubits
\end{enumerate}

If, additionally, $\ket{\psi}$ is the unique ground state of a nearest-neighbor 1D Hamiltonian with spectral gap $\Delta$, then the bound on the bond dimension can be improved to
\begin{enumerate}[(2')]
  \item  $\chi = e^{\tilde{O}\left(\frac{\log^{4}(d)}{\Delta^{2}}\right)} \left(k/\epsilon^2\right)^{O\left(\frac{\log(d)}{\Delta}\right)}$
\end{enumerate}
\end{theorem}

The sort of constant-depth quantum circuit that can generate the state $\kettilde{\psi}$ in Theorem \ref{thm:mainthm} is shown in Figure \ref{fig:constantdepthcircuit}. Proof summaries as well as full proofs of Theorems \ref{thm:improvedbd} and \ref{thm:mainthm} appear in Section \ref{sec:proofs}.

We also note that, unlike Theorem \ref{thm:improvedbd}, Theorem \ref{thm:mainthm} does not require that the chain be longer than some threshold $N_0$; the statement holds regardless of the chain length, although this should be considered a technical detail and not an essential aspect of the constructions. 

\begin{figure}[ht]
    \centering
    \includegraphics[width=\columnwidth]{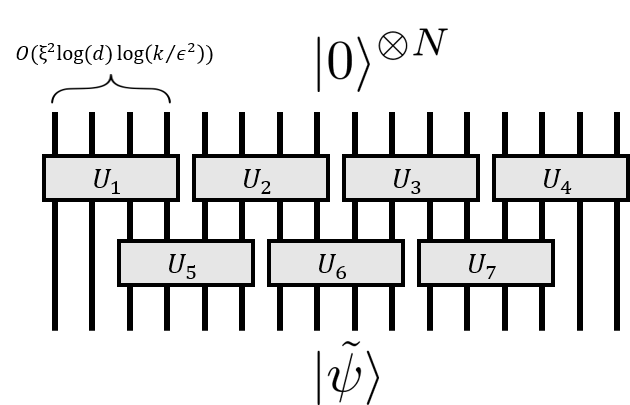}
    \caption{ \label{fig:constantdepthcircuit}Constant-depth quantum circuit that constructs the $(k,\epsilon)$-local approximation $\kettilde{\psi}$ in Theorem \ref{thm:mainthm} starting from the initial state $\ket{0}^{\otimes N}$. It is drawn here as a depth-2 circuit where unitaries $U_j$ act on segments consisting of $O(\xi^2 \log(d) \log(k/{\epsilon^2}))$ contiguous sites. Each of these unitaries could themselves be decomposed into a sequence of nearest-neighbor gates with depth $\text{poly}(k,1/\epsilon)$.}
\end{figure}

\subsection{Reduction from estimating energy density to finding local properties}

The previously stated results show that there exists a state that is both a $(k,\epsilon)$-local approximation and an MPS with bond dimension $\text{poly}(k,1/\epsilon)$. They say nothing of the algorithmic complexity required to \textit{find} a $(k,\epsilon)$-local approximation. The proofs describe how to construct the local approximation from a description of the exact ground state, but following this strategy would require first finding a description of the exact ground state (or perhaps a global approximation to it). One might hope that a different strategy would allow the local approximation to be found much more quickly than the global approximation, since the bond dimension needed to represent the approximation is much smaller. However, the following result challenges the validity of this intuition, at least in the case that the Hamiltonian is translationally invariant, by showing a relationship between the problem of finding a local approximation and the problem of estimating the energy density.

\begin{problem}[Estimating energy density]\label{prob:energydensity}
Given a nearest-neighbor translationally invariant 1D Hamiltonian $H$ on $N\geq 2$ sites and error parameter $\epsilon$, produce an estimate $\tilde{u}$ such that $\lvert u- \tilde{u}\rvert \leq \epsilon$ where $u = E/(N-1)$ is the ground state energy density.
\end{problem}

\begin{problem}[Approximating local properties]\label{prob:localprops}
Given a nearest-neighbor translationally invariant Hamiltonian $H$, an error parameter $\delta$, and an operator $O$  whose support is contained within a contiguous region of length $k$, produce an estimate $\tilde{v}$ such that $\lvert v - \tilde{v} \rvert \leq \delta$, where $v = \bra \psi O \ket \psi/ \lVert O \rVert$ is the expectation value of the operator $O/\lVert O \rVert$ in the ground state $\ket{\psi}$ of $H$. 
\end{problem}

Problem \ref{prob:energydensity} is the restriction of Problem \ref{prob:localprops} to the case where $k=2$ and the operator $O$ is the energy interaction term. Thus, there is a trivial reduction from Problem \ref{prob:energydensity} to Problem \ref{prob:localprops} with $\delta = \epsilon$. However, the next theorem, whose proof is presented in Section \ref{sec:proofs}, states a much more powerful reduction.

\begin{theorem}\label{thm:reduction}
Suppose one has an algorithm that solves Problem \ref{prob:localprops} for any single-site ($k=1$) operator $O$ and $\delta = 0.9$ in $f(\Delta,d,N)$ time, under the promise that the Hamiltonian $H$ has spectral gap at least $\Delta$. Here $d$ denotes the local dimension of $H$ and $N$ the length of the chain. Then there is an algorithm for Problem \ref{prob:energydensity} under the same promise that runs in time 
\begin{equation}
    f\left(\frac{\min(2\Delta,(N-1)\epsilon,2)}{12},2d,N\right)O(\log(1/\epsilon)).
\end{equation} 
\end{theorem}

Estimating the energy density to precision $\epsilon$ is equivalent to measuring the total energy to precision $\epsilon(N-1)$, so the quantity $\min(2\Delta,(N-1)\epsilon,2)$ is equivalent to the global energy resolution, twice the gap, or two, whichever is smallest. Thus, one may take $\epsilon = O(1/N)$ and understand the theorem as stating that finding local properties to within $O(1)$ precision can be done at most $O(\log(N))$ faster than finding an estimate to the total ground state energy to $O(1)$ precision. If local properties can be found in time independent of $N$ (i.e.~there is an $N$-independent upper bound to $f$), then the ground state energy can be estimated to $O(1)$ precision in time $O(\log(N))$, which would be optimal since the ground state energy scales extensively with $N$, and $\Omega(\log(N))$ time would be needed simply to write down the output. 

Another way of understanding the significance of the theorem is in the thermodynamic limit. Here it states that if one could estimate expectation values of local observables in the thermodynamic limit to $O(1)$ precision in some finite amount of time (for constant $\Delta$ and $d$), then one could compute the ground state energy density of such Hamiltonians to precision $\epsilon$ in $O(\log(1/\epsilon))$ time. This would be an exponential speedup over the best-known algorithm for computing the energy density given in \cite{huang2015computing}, which has runtime $\text{poly}(1/\epsilon)$. Taking the contrapositive, if one could show that $\text{poly}(1/\epsilon)$ time is necessary for computing the energy density, this would imply that Problem \ref{prob:localprops} with $\delta = O(1)$ is in general uncomputable in the thermodynamic limit, even given the promise that the input Hamiltonian is gapped. It is already known that Problem \ref{prob:localprops} is uncomputable when there is no such promise \cite{bausch2018undecidability}. 

It is not clear whether a $O(\log(1/\epsilon))$ time algorithm for computing the energy density is possible. The $\text{poly}(1/\epsilon)$ algorithm in \cite{huang2015computing} works even when the Hamiltonian is not translationally invariant, but it is not immediately apparent to us how one might exploit translational invariance to yield an exponential speedup.

\section{Proofs}\label{sec:proofs}
\subsection{Important lemmas for Theorems \ref{thm:improvedbd} and \ref{thm:mainthm}}

The pair of lemmas stated here are utilized in both Theorem \ref{thm:improvedbd} and Theorem \ref{thm:mainthm}. The first lemma captures the essence of the area laws stated previously, and will be essential when we want to bound the error incurred by truncating a state along a certain cut.

\begin{lemma}[Area laws \cite{brandao2015exponential}, \cite{arad2013area}]\label{lem:lowrank}
If $\sigma = \ket{\psi}\bra{\psi}$ has $(t_0, \xi)$-exponential decay of correlations then for any $\chi$ and any region of the chain $A$, there is a state $\tilde{\sigma}_A$ with rank at most $\chi$ such that
\begin{equation}\label{eq:lemmaarealaw}
D(\sigma_A, \tilde{\sigma}_A) \leq C_1\exp\left(-\frac{\log(\chi)}{8\xi\log(d)}\right),
\end{equation}
where $C_1 = \exp(t_0\exp(\tilde{O}(\xi\log(d))))$ is a constant independent of $N$.

If $\sigma$ is the unique ground state of a nearest-neighbor Hamiltonian with spectral gap $\Delta$, then this can be improved to
\begin{equation}\label{eq:lemmabetterarealaw}
D(\sigma_A, \tilde{\sigma}_A) \leq C_2\exp\left(-\tilde{O}\left(\frac{\Delta^{1/3}\log^{4/3}(\chi)}{\log^{4/3}(d)}\right)\right),
\end{equation}
where $C_2 = \exp(\tilde{O}(\log^{8/3}(d)/\Delta))$.
\end{lemma}
\begin{proof}
The first part follows from the main theorem of \cite{brandao2015exponential}. The second follows from the 1D area law presented in \cite{arad2013area}, and $\log(d)$ dependence explicitly stated in \cite{arad2017rigorous}.
\end{proof}

In both proofs we will also utilize the well-known fact that Schmidt ranks cannot differ by more than a factor of $d$ between any two neighboring cuts on the chain.

\begin{lemma}\label{lem:rankrelations}
Any state $\sigma_{AB}$ on a bipartite system $AB$ satisfies the following relations.
\begin{align}
\rank(\sigma_{AB})\rank(\sigma_B) &\geq \rank(\sigma_A) \\
\rank(\sigma_{AB})\rank(\sigma_A) &\geq \rank(\sigma_B) \\
\rank(\sigma_A) \rank(\sigma_B) &\geq \rank(\sigma_{AB})
\end{align}
\end{lemma}
\begin{proof}
We can purify $\sigma_{AB}$ with an auxiliary system $C$ into the state $\ket{\eta}$. We can let $\sigma = \ket{\eta}\bra{\eta}$ and note that $\rank(\sigma_{AB}) = \rank(\sigma_C)$. Thus each of these three equations say the same thing with permutations of $A$, $B$, and $C$. We will show the first equation. Write Schmidt decomposition
\begin{equation}
    \ket{\eta} = \sum_{j=1}^{\rank(\sigma_{AB})} \lambda_j \ket{\nu_j}_{AB} \otimes \ket{\omega_j}_C
\end{equation}
and then decompose $\ket{\nu_j}$ to find
\begin{equation}
    \ket{\eta} = \sum_{j=1}^{\rank(\sigma_{AB})}\sum_{k=1}^{\rank(\sigma_B)} \lambda_j \gamma_{jk}\ket{\tau_{jk}}_A \otimes \ket{\mu_k}_{B} \otimes \ket{\omega_j}_C,
\end{equation}
where $\{\ket{\mu_k}\}_{k=1}^{\rank(\sigma_B)}$ are the eigenvectors of $\sigma_{B}$. This shows that the support of $\sigma_A$ is spanned by the set of $\ket{\tau_{jk}}$ and thus its rank can be at most $\rank(\sigma_{AB})\rank(\sigma_B)$.
\end{proof}

\begin{corollary}\label{cor:schmidtrankincrease}
If $\ket{\eta}$ is a state on a chain of $N$ sites with local dimension $d$, and the Schmidt rank of $\ket{\eta}$ across the cut between sites $m$ and $m+1$ is $\chi$, then the Schmidt rank of $\ket{\eta}$ across the cut between sites $m'$ and $m'+1$ is at most $\chi d^{\lvert m - m' \rvert}$. 
\end{corollary}

\begin{proof}
Without loss of generality, assume $m \leq m'$. The reduced density matrix of $\ket{\eta}\bra{\eta}$ on sites $[m+1,m']$ has rank at most $d^{\lvert m'-m\rvert}$ since this is the dimension of the entire Hilbert space on that subsystem. Meanwhile the rank of the reduced density matrix on sites $[1,m]$ is $\chi$. So by the previous lemma, the rank over sites $[1,m']$ is at most $\chi d^{\lvert m'-m\rvert}$.
\end{proof}


\subsection{Proof of Theorem \ref{thm:improvedbd}}
First we state and prove a lemma that will be essential for showing the first part of Theorem \ref{thm:improvedbd}. Then we provide a proof summary of Theorem \ref{thm:improvedbd}, followed by its full proof.

\begin{lemma}\label{lem:absorbentropy}
Given two quantum systems $A$ and $B$ and states $\tau_A$ on $A$ and $\tau_B$ on $B$, there exists a state $\sigma_{AB}$ on the joint system $AB$ such that $\sigma_A = \tau_A$, $\sigma_B = \tau_B$, and $\rank(\sigma_{AB}) \leq \max(\rank(\tau_A), \rank(\tau_B))$.
\end{lemma}

\begin{proof}
We'll apply an iterative procedure. For round 1 let $\alpha_1 = \tau_A$ and $\beta_1 = \tau_B$. In round $j$ write spectral decomposition \begin{align}
    \alpha_j &= \sum_{i=1}^{a_j} \lambda_{j,i} \ket{s_{j,i}}\bra{s_{j,i}}_A\\
    \beta_j &= \sum_{i=1}^{b_j} \mu_{j,i} \ket{r_{j,i}}\bra{r_{j,i}}_B,
\end{align}
where $a_j$ and $b_j$ are the ranks of states $\alpha_j$ and $\beta_j$, eigenvectors $\{s_{j,i}\}_{i=1}^{a_j}$ and $\{r_{j,i}\}_{i=1}^{b_j}$ form orthonormal bases of the Hilbert spaces of systems $A$ and $B$, respectively, and eigenvalues $\{\lambda_{j,i}\}_{i=1}^{a_j}$ and $\{\mu_{j,i}\}_{i=1}^{b_j}$ are non-decreasing with increasing index $i$ (i.e.~smallest eigenvalues first). Then define 
\begin{equation}\label{eq:uj}
    \ket{u_j} = \sum_{i=1}^{\min(a_j,b_j)} \sqrt{\min(\lambda_{j,i},\mu_{j,i})} \ket{s_{j,i}}_A \otimes \ket{r_{j,i}}_B,
\end{equation}
which may not be a normalized state. Define recursion relation
\begin{align}
    \alpha_{j+1} &= \alpha_j - \Tr_{B}(\ket{u_j}\bra{u_j})\nonumber \\
    \beta_{j+1} &= \beta_j - \Tr_{A}(\ket{u_j}\bra{u_j}) \label{eq:recursion}
\end{align}
and repeat until round $m$ when $\alpha_{m+1} = \beta_{m+1} = 0$. Let
\begin{equation}\label{eq:sigmaAB}
    \sigma_{AB} = \sum_{j=1}^m \ket{u_j} \bra{u_j}.
\end{equation}
Clearly $\rank(\sigma_{AB}) \leq m$. We claim that $m \leq \max(\rank(\tau_A), \rank(\tau_B))$. To show this we note that
\begin{equation}\label{eq:rankdecrease}
    a_{j+1} + b_{j+1} \leq \max(a_j, b_j).
\end{equation}
We can see this is true by inspecting the $i$th term in the Schmidt decomposition of $\ket{u_j}$ in Eq.~\eqref{eq:uj}, and noting that either its reduced density matrix on system $A$ is $\lambda_{j,i}\ket{s_{j,i}} \bra{s_{j,i}}$ or its reduced density matrix on system $B$ is $\mu_{j,i}\ket{r_{j,i}} \bra{r_{j,i}}$ (or both). So when the reduced density matrices of $\ket{u_j}\bra{u_j}$ are subtracted from $\alpha_j$ and $\beta_j$ to form $\alpha_{j+1}$ and $\beta_{j+1}$ in Eqs.~\eqref{eq:recursion}, each of the $\min(a_j, b_j)$ terms causes the combined rank $a_{j+1} + b_{j+1}$ to decrease by at least one in comparison to $a_j+b_j$.

This alone implies that $m \leq \max(a_1, b_1)+1$, since by Eq.~\eqref{eq:rankdecrease}, the sequence $\{a_j+b_j\}_j$ must decrease by at least $\min(a_1,b_1)$ after the first round, and then by at least $1$ in every other round, reaching 0 when $j = \max(a_1,b_1)+1$. However, we can also see that the last round must see a decrease by at least 2, because it is impossible for $a_{m} = 0$ and $b_m =1$ or vice versa (since $\Tr(\alpha_j)$ must equal $\Tr(\beta_j)$ for all $j$). Thus $m \leq \max(a_1,b_1)$.

Moreover, Eqs.~\eqref{eq:recursion} and \eqref{eq:sigmaAB} imply that $\sigma_A = \alpha_1 = \tau_A$ and $\sigma_B = \beta_1 = \tau_B$. 
\end{proof}

\begin{figure}[ht]
\includegraphics[width=\linewidth]{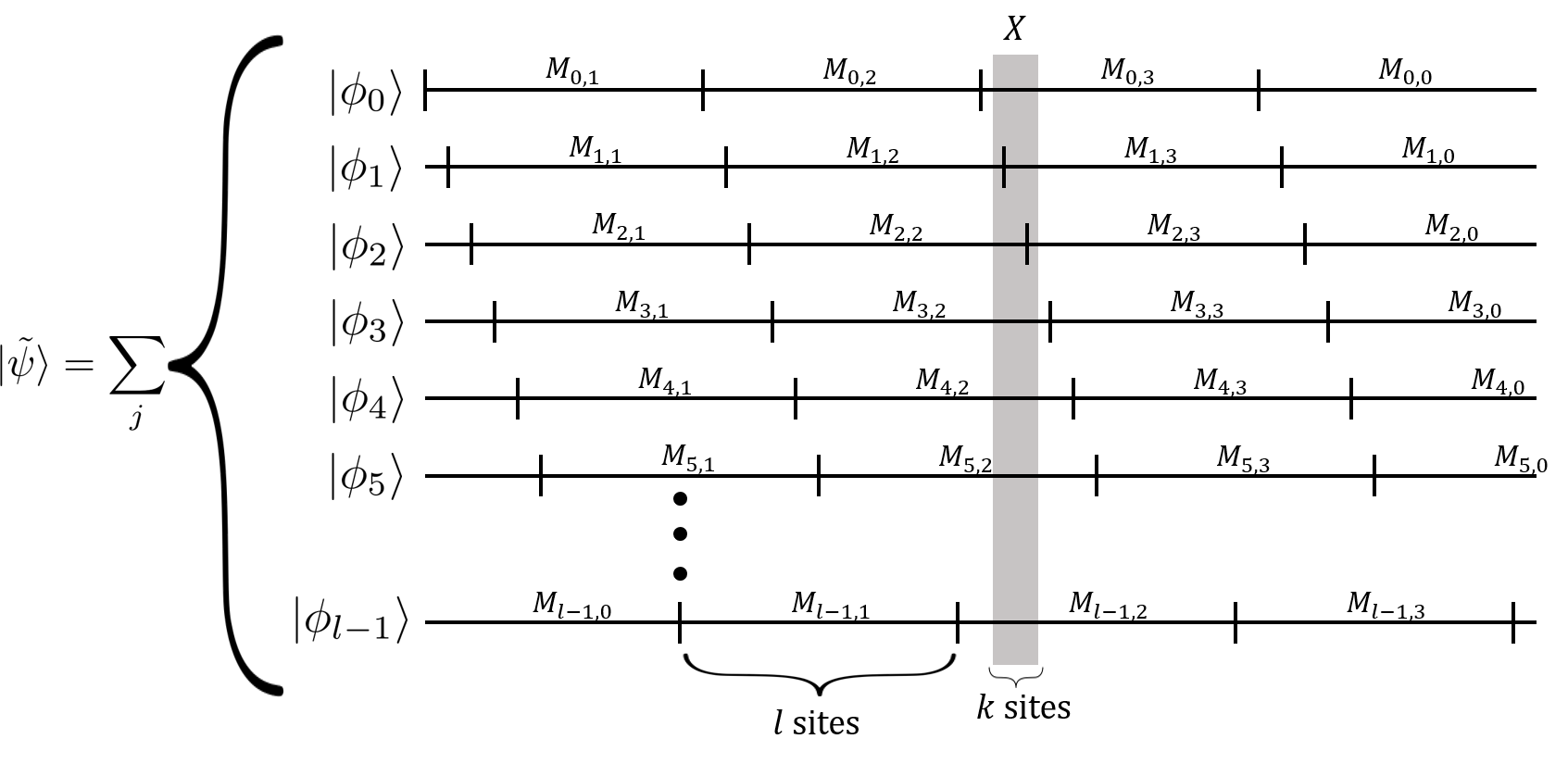}

\caption{\label{fig:improvedbd}Schematic overview of the proof of Theorem \ref{thm:improvedbd}. Many states $\ket{\phi_j}$ are constructed with staggered divisions between regions $M_{j,i}$ of length $l$, then the $\ket{\phi_j}$ are summed in superposition. Properties supported within a length-$k$ region $X$ are faithfully captured by $\ket{\phi_j}$ for values of $j$ such that $X$ does not overlap the boundaries between regions $M_{j,i}$. Most values of $j$ qualify under this criteria as long as $l$ is much larger than $k$. Additional structure is defined (the regions $B_{j,j'}$ in Part 1, and $B_{j,i}$ in Part 2) in order to force $\braket{\phi_j}{\phi_{j'}} =\delta_{jj'}$, but this structure is not reflected on the schematic.}
\end{figure}

\begin{proof}[Proof summary of Theorem \ref{thm:improvedbd}]

In Ref.~\cite{schuch2017matrix}, an MPO that is a $(k,\epsilon)$-local approximation to a given state $\ket{\psi}$ was formed by dividing the chain into many length-$l$ segments, tensoring together a low-bond-dimension approximation of the exact state reduced to each segment, and then summing over (as a mixture) translations of locations for the divisions between the segments. We follow the same idea but for pure states: for each integer $j = 0,\ldots, l-1$, we divide the state into many length-$l$ segments and create a pure state approximation $\ket{\phi_j}$ that captures any local properties that are supported entirely within one of the segments. Then to form $\kettilde{\psi}$, we sum in superposition over all the $\ket{\phi_j}$, where each $\ket{\phi_j}$ has boundaries between segments occurring in different places (see Figure \ref{fig:improvedbd}). Thus, for any length-$k$ region $X$, a large fraction of the terms $\ket{\phi_j}$ in the superposition are individually good approximations on region $X$. The fact that a small fraction of the terms are not necessarily a good approximation creates additional, but small, error in the local approximation.

In order to avoid interaction between different terms in the superposition, we add additional structure to make the $\ket{\phi_j}$ in the superposition exactly orthogonal to one another. In our construction for general states, this additional structure consists of a set of disjoint, sparsely distributed, single-site regions $B_{j,j'}$, one for each pair of integers $j \neq j'$ with $0\leq j, j' < l$. We force $\ket{\phi_j}$ to be the pure state $\ket{0}$ and $\ket{\phi_{j'}}$ to be $\ket{1}$ when reduced to $B_{j,j'}$ to guarantee that $\braket{\phi_j}{\phi_{j'}}=0$. Our construction for states that have exponential decay of correlations is similar: we define a series of regions $B_{j,i}$ and for each pair $(j,j')$ force $\ket{\phi_j}$ and $\ket{\phi_{j'}}$ to have orthogonal supports when reduced to one of these regions. 

Our approach for constructing $\ket{\phi_j}$ differs if $\ket{\psi}$ is a general state (Part 1), or if it is a state that either has exponentially decaying correlations or (additionally) is the ground state of a nearest-neighbor Hamiltonian (Part 2). In the latter case, we examine each length-$l$ segment individually and truncate the bonds of the exact state on all but a few of the rightmost sites within that segment. The area law implies these truncations have minimal effect. We use those few rightmost sites to purify the mixed state on the rest of the segment. Then $\ket{\phi_j}$ is a tensor product over pure states on each of the segments. The bond dimension can be bounded within each segment of $\ket{\phi_j}$ which is sufficient to bound the bond dimension of $\kettilde{\psi}$.

This does not work for general states because without an area law, bond truncations result in too much error, and without the truncations we do not have enough room to purify the state. For a general state, there is simply too much entropy in the length-$l$ segment to fully absorb with only a few sites at the edge of the segment.  Instead, we have the various segments absorb each other's entropy by developing a procedure to engineer entanglement between different length-$l$ segments that exactly preserves the reduced density matrix on each segment and keeps the Schmidt rank constant (albeit exponential in $k/\epsilon$) across any cut. The crux of this procedure is captured in Lemma \ref{lem:absorbentropy}.
\end{proof}

\begin{proof}[Proof of Theorem \ref{thm:improvedbd}]

\textit{Part 1: Items (1) and (2)}
\\
First we consider the case of a general state $\ket{\psi}$ and construct an approximation $\kettilde{\psi}$ satisfying items (1) and (2). A different construction is given in Part 2 to show items (2') and (2''), but it is similar in approach. Throughout this proof, we use a convention where sites are numbered $0$ to $N-1$, which differs from the rest of the paper.

We choose integer $l$ to be specified later. We require $l > k$. To construct $\kettilde{\psi}$, we will first construct states $\ket{\phi_j}$ for $j=0,\ldots, l-1$ and then sum these states in superposition. Any reference to the index $j$ will be taken mod $l$. Consider a fixed value of $j$. To construct $\ket{\phi_j}$ we break the chain into $n = \lfloor N/l \rfloor$ blocks where $n-1$ blocks $\{M_{j,i}\}_{i=1}^{n-1}$ have length exactly $l$ and the final block $M_{j,0}$ has length at least $l$ and less than $2l$. We arrange the blocks so that the leftmost site of block $M_{j,1}$ is site $j$; thus block $M_{j,i}$ contains the sites $[j+l(i-1),j+li-1]$ for $i=1,\ldots, n-1$, and the last block $M_{j,0}$ ``rolls over'' to include sites at both ends of the chain: sites $[j+l(n-1),N-1]$ and $[0,j-1]$. A schematic of the arrangement is shown in Figure \ref{fig:improvedbd}. Any reference to the index $i$ will be taken mod $n$.

We also define $l^2-l$ single-site blocks on the chain which we label $B_{j,j'}$ for all pairs $j \neq j'$. $B_{j,j'}$ consists of the site with index $3l^2(j+j')+l(j-j')+2l^2$. This definition is possible as long as their are sufficiently many sites: $N \geq N_0 = O(l^3)$. It can be verified that since $0 \leq j,j' \leq l-1$, the distance between any $B_{j,j'}$ and $B_{j'',j'''}$ for any distinct pairs $(j,j')$ and $(j'',j''')$ is at least $l$. For each $j$, let $B_{j} = \cup_{j' \neq j} B_{j,j'} \cup B_{j',j}$. For each $j,i$, let $A_{j,i} = M_{j,i} \setminus M_{j,i}\cap B_j$. Thus, in most cases $A_{j,i} = M_{j,i}$ since $B_j$ is a relatively small set of sites. Let $A_j = \cup_{i=0}^{n-1} A_{j,i}$ = $B_j^c$, the complement of $B_j$.

The state $\ket{\phi_j}$ will have the form
\begin{equation}\label{eq:phij1}
    \ket{\phi_j} = \ket{Q_{j}}_{A_j} \otimes \bigotimes_{j'\neq j}\left( \ket{0}_{B_{j,j'}}\otimes  \ket{1}_{B_{j',j}}\right),
\end{equation}
where $\ket{0}$ and $\ket{1}$ are two of the $d$ computational basis states located on a single site. In other words $\ket{\phi_j}$ is a product state over all single site regions $B_{j,j'}$ with some other (yet to be specified) state $\ket{Q_j}$ on the remainder of the chain $A_j$. 

To construct $\ket{Q_j}$, we apply Lemma \ref{lem:absorbentropy} iteratively as follows. We let $\sigma_1 = \rho_{A_{j,j+2}}$ (the reduced matrix of the exact state $\rho$ on region $A_{j,j+2}$). We combine $\sigma_1$ with $\rho_{A_{j,j+3}}$ using Lemma \ref{lem:absorbentropy} to form a state $\sigma_2$ on region $A_{j,j+2}A_{j,j+3}$ such that $\rank(\sigma_2) \leq \max( \rank(\sigma_1), \rank(\rho_{A_{j,j+3}}))$. For any $j,i$, the rank of the state on any region $A_{j,i}$ is less than $d^{2l}$ since any region contains at most $2l$ sites. So if we apply this process iteratively, forming $\sigma_{p+1}$ by combining $\sigma_p$ and the state on region $A_{j,j+p+2}$ ($j+p+2$ is taken mod $n$), then we end up with a state $\sigma_{n-2}$ with rank at most $d^{2l}$ defined over all of $A_j$ except $A_{j,j}$ and $A_{j,j+1}$. Since by construction $B_j$ contains no sites with index smaller than $2l^2$ and $A_{j,j}$ and $A_{j,j+1}$ are contained within the first $2l^2$ sites, we have $A_{j,j}=M_{j,j}$ and $A_{j,j+1} = M_{j,j+1}$ meaning each of these two regions each contain $l$ sites and the total dimension of the Hilbert space over $A_{j,j}A_{j,j+1}$ is at least $d^{2l}$. Thus we may use regions $A_{j,j}$ and $A_{j,j+1}$ to purify the state $\sigma_{n-2}$. We let $\ket{Q_j}$ be any such purification.

The key observation is that the state $\ket{\phi_j}$, as defined by Eq.~\eqref{eq:phij1}, will get any local properties exactly correct as long as they are supported entirely within a segment $A_{j,i}$ for some $i \neq j, j+1$. As long as $l$ is large, this will be the case for most length-$k$ regions, but it will not be the case for some regions that cross the boundaries between regions $A_{j,i}$ or for regions that contain one of the single site regions $B_{j,j'}$ or $B_{j',j}$. 

To fix this we sum in \textit{superposition} over the states $\ket{\phi_j}$ for each value of $j$. The motivation to do this is so that every length-$k$ region will be contained within $A_{j,i}$ for some value of $i$ in most, but not all, of the terms in the superposition. We will show that most is good enough. We let
\begin{equation}
    \kettilde{\psi} = \frac{1}{\sqrt{l}}\sum_{j=0}^{l-1}\ket{\phi_j}.
\end{equation}
We note that $\braket{\phi_j}{\phi_{j'}} = \delta_{jj'}$ since $\ket{\phi_j}$ is simply $\ket{0}$ when reduced to region $B_{j,j'}$ and $\ket{1}$ when reduced to region $B_{j',j}$, while $\ket{\phi_{j'}}$ is $\ket{1}$ when reduced to region $B_{j,j'}$ and $\ket{0}$ when reduced to $B_{j',j}$. Thus $\kettilde{\psi}$ is normalized:
\begin{equation}
    \langle \tilde{\psi} | \tilde{\psi} \rangle = \frac{1}{l}\sum_{j,j'} \braket{\phi_j}{\phi_{j'}} = 1.
\end{equation}

This completes the construction of the approximation. We now wish to show it has the desired properties. To show item (1), we compute the (local) distance from $\ket{\psi}$ to $\kettilde{\psi}$. Let $\tilde{\rho} = \kettilde{\psi}\bratilde{\psi}$ and consider an arbitrary length-$k$ region $X$. We may write
\begin{align}
   & D_1(\rho_X, \tilde{\rho}_X) \nonumber\\
   ={}&\frac{1}{2}\left\lVert\left(\frac{1}{l} \sum_{j=0}^{l-1}\sum_{j'=0}^{l-1}\Tr_{X^c}(\ket{\phi_j}\bra{\phi_{j'}})\right)-\Tr_{X^c}(\ket{\psi}\bra{\psi})\right\rVert_1. 
\end{align}
First we examine terms in the sum for which $j \neq j'$. Since $B_{j,j'}$ and $B_{j',j}$ are separated by at least $l$ sites and $l > k$, $X^c$ must include either $B_{j,j'}$ or $B_{j',j}$ (or both). Since $\ket{\phi_j}$ and $\ket{\phi_{j'}}$ have orthogonal support on both those regions, and at least one of them is traced out, the term vanishes. 

Thus we have
\begin{align}
    D_1(\rho_X, \tilde{\rho}_X) &= \frac{1}{2}\left\lVert\frac{1}{l} \sum_{j=0}^{l-1}\Tr_{X^c}(\ket{\phi_j}\bra{\phi_{j}}-\ket{\psi}\bra{\psi})\right\rVert_1 \nonumber\\
    &\leq\frac{1}{l} \sum_{j=0}^{l-1}D_1\left(\rho_X, \Tr_{X^c}(\ket{\phi_j}\bra{\phi_j})\right). \label{eq:tracesumj2}
\end{align}

For a particular $j$, there are two cases. Case 1 includes values of $j$ for which $X$ falls completely within the $A_{j,i}$ for some $i$ with $i \neq j, j+1$. For these values of $j$ the term vanishes because the reduced density matrix of $\ket{\phi_j}\bra{\phi_j}$ on $X$ is exactly $\rho_X$. Case 2 includes all other values of $j$. For this to be the case, either $X$ spans the boundary between two regions $M_{j,i}$ and $M_{j,i+1}$ (at most $k-1$ different values of $j$), $X$ contains a site $B_{j,j'}$ or $B_{j',j}$ for some $j'$ (at most 2 values of $j$, since the separation between sites $B_{j,j'}$ implies only one may lie within $X$), or $X$ is contained within $A_{j,j}$ or $A_{j,j+1}$ (at most 2 values of $j$). In this case, the term will not necessarily be close to zero, but we can always upper bound the trace distance by 1. The number of terms in the sum that qualify as Case 2 is therefore at most $k+3$, and the total error can be bounded:
\begin{equation}
   D_1(\rho_X, \tilde{\rho}_X)  \leq \frac{k+3}{l}.
\end{equation}
Choosing $l = (k+3)/\epsilon$ shows item (1), that $\kettilde{\psi}$ is a $(k,\epsilon)$-local approximation to $\ket{\psi}$.

To show item (2), we bound the Schmidt rank of the state $\kettilde{\psi}$ across every cut that bipartitions the chain into two contiguous regions. Since $\kettilde{\psi}$ is a superposition over $l$ terms $\ket{\phi_j}$, the Schmidt rank can be at most $l$ times greater than that of an individual $\ket{\phi_j}$. Fix some value of $j$ and some cut of the chain at site $s$. Since $\ket{\phi_j}$ is a product state between regions $A_j$ and $B_j$, and moreover it is a product state on each individual site in $B_j$, we may ignore $B_j$ when calculating the Schmidt rank (it has no entanglement), and focus merely on $\ket{Q_j}_{A_j}$. We constructed the state $\ket{Q_j}$ by building up mixed states $\sigma_p$ on region $A_{j,j+2}\ldots A_{j,j+p+1}$ until $p = n-2$, then purified with the remaining two regions. Each $\sigma_p$ has $\rank(\sigma_p) \leq d^{2l}$. Now consider an integer $b$ with $1 \leq b \leq n-1$ and $b \neq j,j+1$. Denote $\sigma = \sigma_{n-2}$ and note that
\begin{align}
    &\rank(\sigma_{A_{j,1}\ldots A_{j,b}}) \nn
    \leq{}&\rank(\sigma_{A_{j,j+2}\ldots A_{j,b}})\rank(\sigma_{A_{j,j+2}\ldots A_{j,0}}) \nn
    ={}&\rank(\sigma_{b-j-1})\rank(\sigma_{n-j-1})
    \leq d^{4l},
\end{align}
where the first inequality follows from Lemma \ref{lem:rankrelations}.

Moreover, we may choose $b$ such that the cut between $A_{j,b}$ and $A_{j,b+1}$ falls within $l$ sites of site $s$. We find that the region $A_{j,1}\ldots A_{j,b}$ differs from the region containing sites $[0,s]$ by at most $l$ sites at each edge. Thus the Schmidt rank on the region left of the cut can be at most $d^{2l}$ larger than that of $A_{j,1}\ldots A_{j,b}$ (Corollary \ref{cor:schmidtrankincrease}), giving a bound of $d^{6l}$ for the the Schmidt rank of $\ket{\phi_j}$.  This implies the Schmidt rank of $\kettilde{\psi}$ is at most $ld^{6l}$, which proves item (2). This applies whenever $N \geq N_0 = O(l^3) = O(k^3/\epsilon^3)$, a bound which must be satisfied in order for the chain to be long enough to fit all the regions $B_{j,j'}$ as defined above. This completes Part 1.


\vspace{12 pt}

\textit{Part 2: Items (2') and (2'')} \\

This construction is mostly similar to the previous one with a few key differences. We choose integers $l$, $t$, and $\chi'$ to be specified later. We require $t$ be even and $l \geq 2k$, $l \geq 2t$. We assume $N \geq N_0 = 2l^2$. We also require that $d \geq 4$. If this is not the case, we coarse-grain the system by combining neighboring sites, and henceforth we assume $d \geq 4$.

As in Part 1, to construct $\kettilde{\psi}$, we will first construct states $\ket{\phi_j}$ for $j=0,\ldots, l-1$ and then sum these states in superposition. Consider a fixed value of $j$. To construct $\ket{\phi_j}$ we break the chain into $n = \lfloor N/l \rfloor$ blocks and we arrange them exactly as in Part 1.

Now the construction diverges from Part 1: the state $\ket{\phi_j}$ will be a product state over each of these blocks
\begin{equation}
    \ket{\phi_j} = \ket{\phi_{j,0}}_{M_{j,0}} \otimes \ldots \otimes \ket{\phi_{j,n-1}}_{M_{j,n-1}}
\end{equation}
with states $\ket{\phi_{j,i}}$ for $i=0,\ldots, n-1$ that we now specify.

The idea is to create a state $\ket{\phi_{j,i}}$ that has nearly the same reduced density matrix as $\ket{\psi}$ on the leftmost $l-t$ sites of region $M_{j,i}$. It uses the rightmost $t$ sites to purify the reduced density matrix on the leftmost $l-t$ sites. First we denote the leftmost $l-t$ sites of block $M_{j,i}$ by $A_{j,i} = [j+l(i-1),j+li-t-1]$ and the rightmost $t$ sites by $B_{j,i} = [j+li-t,j+li-1]$ (or appropriate ``roll over'' definitions when $i=0$). We write $\ket{\psi}$ as an exact MPS with exponential bond dimension and form $\ket{\psi_{j,i}}$ by truncating to bond dimension $\chi'$ (i.e.~projecting onto the span of the right Schmidt vectors associated with the largest $\chi'$ Schmidt coefficients, then normalizing the state) at the cut to the left of region $A_{j,i}$, every cut within $A_{j,i}$, and at the cut to the right of $A_{j,i}$, for a total of at most $2l-t$ truncations (recall the final region $M_{j,0}$ may have as many as $2l-1$ sites). We denote the pure density matrix of this state by $\rho^{(j,i)} = \ket{\psi_{j,i}}\bra{\psi_{j,i}}$. We can bound the effect of these truncations using the area law given by Lemma \ref{lem:lowrank}:
\begin{equation}
    D(\rho^{(j,i)}, \rho) \leq \sqrt{2l-t}\epsilon^{\chi'},
\end{equation}
where $\epsilon^{\chi'}$ is the cost (in purified distance) of a single truncation, given by the right hand side of Eq.~\eqref{eq:lemmaarealaw}, or by Eq.~\eqref{eq:lemmabetterarealaw} in the case $\ket{\psi}$ is the ground state of a gapped nearest-neighbor 1D Hamiltonian. These truncations were not possible in Part 1 because we could not invoke the area law for general states.

Because of the truncations, we can express the reduced density matrix $\rho_{A_{j,i}}^{(j,i)}$ as a mixture of $\chi'^2$ pure states $\{\ket{\phi_{j,i,z}}\}_{z=0}^{\chi'^2-1}$ each of which can be written as an MPS with bond dimension $\chi'$: %
\begin{equation}
    \rho_{A_{j,i}}^{(j,i)} = \sum_{z=0}^{\chi'^2-1}p_{j,i,z} \ket{\phi_{j,i,z}}\bra{\phi_{j,i,z}},
\end{equation}
for some probability distribution $\{p_{j,i,z}\}_{z=0}^{\chi'^2-1}$. We now form $\ket{\phi_{j,i}}$ by purifying $\rho_{A_{j,i}}^{(j,i)}$ onto the region $M_{j,i}$ using the space $B_{j,i}$, which contains $t$ sites, as the purifying subspace:
\begin{equation}
    \ket{\phi_{j,i}}_{M_{j,i}}= \sum_{z=0}^{\chi'^2-1}\sqrt{p_{j,i,z}} \ket{\phi_{j,i,z}}_{A_{j,i}} \otimes \ket{r_{j,i,z}}_{B_{j,i}},
\end{equation}
where the set of states $\{\ket{r_{j,i,z}}\}_{z=0}^{\chi'^2-1}$ is an orthonormal set defined on region $B_{j,i}$. This purification will only be possible if the dimension of $B_{j,i}$ is sufficiently large, and we comment later on this fact, as well as how exactly to choose the set $\{\ket{r_{j,i,z}}\}_{z=0}^{\chi'^2-1}$.

The key observation is that state $\ket{\phi_j}$ will get any local properties approximately correct as long as they are supported entirely within a segment $A_{j,i}$ for some $i$, and as long as $\chi'$ is large enough that the $2l-t$ truncations do not have much effect on the reduced density matrix there. Thus we will choose our parameters so that $\chi'$ is large (but independent of $N$), such that $l$ is much larger than $t$ and $k$ (so that most regions fall within a region $A_{j,i}$), and such that $t$ is large enough that it is possible to purify states on $A_{j,i}$ onto $M_{j,i}$. But, as in Part 1, we have the issue that some regions will not be contained entirely within region $A_{j,i}$ for some $i$. 

We again deal with this issue by summing in superposition:
\begin{equation}
    \kettilde{\psi} = \frac{1}{\sqrt{l}}\sum_{j=0}^{l-1}\ket{\phi_j}.
\end{equation}

To complete the construction we also must specify the orthonormal states $\{\ket{r_{j,i,z}}\}_{z=0}^{\chi'^2-1}$ defined on the $t$ sites in region $B_{j,i}$. We choose a set that satisfies the following requirements.

\begin{enumerate}[(1)]
    \item The reduced density matrix of $\ket{r_{j,i,z}}$ on any single site among the leftmost $t/2$ sites of $B_{j,i}$ (recall we have assumed $t$ is even) is entirely supported on basis states $1, \ldots, \lfloor d/2 \rfloor$.
    \item The reduced density matrix on any single site among the rightmost t/2 sites is entirely supported on basis states $\lfloor d/2 \rfloor +1, \ldots, d$.
    \item Let $j' = j+i \mod l$, and let $i_0 = i \mod l$. If $t \leq i_0 \leq l-t$ then for all $z$, $\ket{r_{j,i,z}}$ is orthogonal to the support of the reduced density matrix of $\ket{\phi_{j'}}$ on region $B_{j,i}$.
\end{enumerate}

We assess how large $t$ must be for it to be possible to satisfy these three conditions. The third item specifically applies only for values of $i$ that lead to values of $j'$ that are at least $t$ away from $j$ (modulo $l$) so that the purifying system $B_{j,i}$ does not overlap with $B_{j',i'}$ for any $i'$. The support of the reduced density matrix of any $\ket{\phi_{j'}}$ on region $B_{j,i}$ has dimension at most $\chi'^2$. Thus, if the dimension of $B_{j,i}$ is more than $2\chi'^2$ it will always be possible to choose an orthonormal set $\{\ket{r_{j,i,z}}\}_{z=0}^{\chi'^2-1}$ satisfying the third condition. The first and second conditions cut the accessible part of the local dimension of the purifying system in half, so a purification that satisfies all three conditions will be possible if $\lfloor d/2 \rfloor^t \geq 2\chi'^2$. Any choice of set that meets all three conditions is equally good for our purposes.

We now demonstrate that the three conditions imply that for any pair $(j,j')$ there are regions of the chain on which the supports of the reduced density matrices of states $\ket{\phi_j}$ and $\ket{\phi_{j'}}$ are orthogonal. If it is the case that $j-j' \mod l < t$ or $j'-j \mod l < t$, then for every $i$ the region $B_{j,i}$ overlaps with $B_{j',i'}$ for some $i'$. Because $B_{j,i} \neq B_{j',i'}$, there will be some site that is in the right half of one of the two regions, but in the left half of the other, and items 1 and 2 imply that the two states will be orthogonal when reduced to this site. If this is not the case, then as long as there is some $i$ for which $j' = j+i \mod l$, then item 3 implies the orthogonality of the supports of $\ket{\phi_j}$ and $\ket{\phi_{j'}}$. In fact because $n \geq 2l$, there will be at least 2 such values of $i$. We conclude that $\braket{\phi_j}{\phi_{j'}} = \delta_{jj'}$, which implies that $\kettilde{\psi}$ is normalized as shown by the computation
\begin{equation}
    \langle \tilde{\psi} | \tilde{\psi} \rangle = \frac{1}{l}\sum_{j,j'} \braket{\phi_j}{\phi_{j'}} = 1.
\end{equation}

We have now shown how to define the approximation $\kettilde{\psi}$ and discussed the conditions for the parameters $t$, $\chi'$, and $d$ that make the construction possible. Now we assess the error in the approximation (locally). Let $\tilde{\rho} = \kettilde{\psi}\bratilde{\psi}$ and consider an arbitrary length-$k$ region $X$. We may write
\begin{align}
   & D_1(\rho_X, \tilde{\rho}_X) \nonumber\\
   ={}&\frac{1}{2}\left\lVert\left(\frac{1}{l} \sum_{j=0}^{l-1}\sum_{j'=0}^{l-1}\Tr_{X^c}(\ket{\phi_j}\bra{\phi_{j'}})\right)-\Tr_{X^c}(\ket{\psi}\bra{\psi})\right\rVert_1.
\end{align}
For the same reason that led to the conclusion $\braket{\phi_j}{\phi_{j'}} = \delta_{jj'}$, we can conclude that $\Tr_{X^c}(\ket{\phi_j}\bra{\phi_{j'}}) = \delta_{jj'}$, so long as $k$ is smaller than $l/2$. To see that this holds, it is sufficient to show that there is a region lying completely outside of $X$ with the property that $\ket{\phi_j}$ and $\ket{\phi_j'}$ share no support on the region. Since $k  = \lvert X \rvert \leq l/2$ and $\lvert B_{j,i}\rvert = t \leq l/2 $, for any $j$, $X$ can overlap the region $B_{j,i}$ for at most one value of $i$. We showed before that for any $j$ there would be at least two values of $i$ for which a subregion of $B_{j,i}$ has this property, implying one of them must lie outside $X$.

Thus we have
\begin{align}
    D_1(\rho_X, \tilde{\rho}_X) &= \frac{1}{2}\left\lVert\frac{1}{l} \sum_{j=0}^{l-1}\Tr_{X^c}(\ket{\phi_j}\bra{\phi_{j}}-\ket{\psi}\bra{\psi})\right\rVert_1 \nonumber\\
    &\leq\frac{1}{l} \sum_{j=0}^{l-1}D_1\left(\rho_X, \Tr_{X^c}(\ket{\phi_j}\bra{\phi_j})\right). \label{eq:tracesumj}
\end{align}

For a particular $j$, there are two cases. Case 1 occurs if $X$ falls completely within the $A_{j,i}$ for some $i$, in which case the only error is due to the $2l-t$ truncations to bond dimension $\chi'$. Since the trace norm $D_1$ is smaller than the purified distance $D$ (Lemma \ref{lem:purifiedvstrace}), the contribution for these values of $j$ is at most $\sqrt{2l-t}\epsilon^{\chi'}$. Case 2 includes values of $j$ for which $X$ does not fall completely within a region $A_{j,i}$ for any $i$. In this case, the term will not necessarily be close to zero, but we can always upper bound the trace distance by 1. The number of terms in the sum that qualify as Case 1 is of course bounded by $l$ (there are only $l$ terms) and the number of Case 2 terms is at most  $t+k-1$. Thus the total error can be bounded:
\begin{align}
   & D_1(\rho_X, \tilde{\rho}_X)  \nonumber \\
    \leq{}&\frac{1}{l} (l\sqrt{2l-t}\epsilon^{\chi'} + (t+k-1)) \nonumber\\
    \leq{} & \sqrt{2l} \epsilon^{\chi'} + (k+t)/l.
\end{align}

We will choose parameters so this quantity is less than $\epsilon$. Parameters $l$ and $t$ will be related to $\chi'$ as follows.
\begin{align}
    t &= \log(2\chi'^2)/\log(\lfloor d/2\rfloor)\\
    l &= 2(k+t)/\epsilon
\end{align}
If $\ket{\psi}$ is known only to have exponentially decaying correlations, then we choose
\begin{equation}
    \log(\chi') = 16 \xi \log(d)\log(16C_1\sqrt{ \xi \log(d)(k+3)}/\epsilon^{3/2}),
\end{equation}
where $C_1 = \exp(t_0\exp(\tilde{O}(\xi\log(d))))$ is the constant from Eq.~\eqref{eq:lemmaarealaw}. We note that $t \leq 3\log(\chi')$, so we can bound
\begin{align}
        D_1(\rho_X, \tilde{\rho}_X) &\leq  \frac{\sqrt{4(k+t)}C_1}{\sqrt{\epsilon}}e^{-\frac{\log(\chi')}{8\xi \log(d)}}+\frac{\epsilon}{2} \nonumber \\
        &\leq \frac{\sqrt{4(k+3)\log(\chi')}C_1}{\sqrt{\epsilon}}e^{-\frac{\log(\chi')}{8\xi \log(d)}}+\frac{\epsilon}{2} \nonumber \\
        &\leq \frac{\sqrt{64\xi\log(d)(k+3)}C_1}{\sqrt{\epsilon}}e^{-\frac{\log(\chi')}{16\xi \log(d)}}+\frac{\epsilon}{2} \nonumber \\
        &\leq \frac{\epsilon}{2} + \frac{\epsilon}{2}  = \epsilon,
\end{align}
where in the third line we have used the (crude) bound $\sqrt{u} \leq e^u$ with $u = \log(\chi')/(16\xi\log(d))$.

In the case that $\ket{\psi}$ is known to be the ground state of a gapped local Hamiltonian, we may choose
\begin{equation}
    \log(\chi') = \tilde{O}(\Delta^{-1/4} \log(d)\log^{3/4}(C_2 \sqrt{k+3}/\epsilon^{3/2})),
\end{equation}
where $C_2 = \exp(\tilde{O}(\log^{8/3}(d)/\Delta)$ is the constant in Eq.~\eqref{eq:lemmabetterarealaw} and the same analysis will follow. This proves item (1) of the theorem for the construction in Part 2.

Items (2') and (2'') assert that $\kettilde{\psi}$ can be written as an MPS with constant bond dimension, which we now show. Each state $\ket{\phi_j}$ is a product state with pure state $\ket{\phi_{j,i}}$ on each each block $M_{j,i}$, and $\ket{\phi_{j,i}}$ has bond dimension at most $2\chi'^2$. Thus, if we cut the state $\ket{\phi_j}$ at a certain site, the bond dimension will be at most $4 \chi'^2$ (recall that block $M_{j,0}$ may have sites at both ends of the chain and can contribute to the bond dimension). Since $\kettilde{\psi}$ is a sum over $\ket{\phi_j}$ for $l$ values of $j$, the bond dimension $\chi$ of $\kettilde{\psi}$ is at most $4l \chi'^2$. For the case of exponentially decaying correlations, this evaluates to
\begin{align}
    \chi &= 4l(256C_1^2 \xi \log(d)(k+3)/\epsilon^3)^{16\xi\log(d)} \nonumber\\
    &\leq e^{t_0e^{\tilde{O}(\xi\log(d))}} (k/\epsilon^3)^{O(\xi\log(d))},
\end{align}
proving item (2'), and for the case of ground state of gapped Hamiltonian, we find
\begin{align}
    \chi &= 4l\exp(\tilde{O}(\Delta^{-1/4} \log(d)\log^{3/4}(C_2\sqrt{k+3}/\epsilon^{3/2}))) \nonumber\\
    &\leq (k/\epsilon)e^{\tilde{O}\left(\frac{\log^3(d)}{\Delta}+\frac{\log(d)}{\Delta^{1/4}}\log^{3/4}(k/\epsilon^3)\right)},
\end{align}
proving item (2''), where the second factor is asymptotically $(k/\epsilon)^{o(1)}$. For completeness, we note that if we combined neighboring sites because $d < 4$, we can now uncombine them possibly incurring a factor of 2 or 3 increase in the bond dimension, which has no effect on the stated asymptotic forms for $\chi$.

These results hold as long as $N \geq 2 l^2$, which translates to $N \geq O(k^2/\epsilon^2) + t_0\exp(\tilde{O}(\xi \log (d)))$ in the case of exponentially decaying correlations, and $N \geq O(k^2/\epsilon^2) + \tilde{O}(\log(d)/\Delta^{3/4})$ in the case that $\ket{\psi}$ is a ground state of a local Hamiltonian. This completes the proof.
\end{proof}

Now we demonstrate that the state $\kettilde{\psi}$ constructed in the proof of Theorem \ref{thm:improvedbd}, Part 2, is long-range correlated. Given an integer $m$, consider the pair of regions $A = [0,l(l-1)-1]$ and $C=[l(l-1+m),N-1]$, which are separated by $ml$ sites. Assume $n \geq 2l+m$, so that $A$ and $C$ both contain at least $l^2$ sites. Define the following operators.
\begin{align}
Q_1 =& \ket{\phi_{0,1}}\bra{\phi_{0,1}}\otimes \ldots \otimes  \ket{\phi_{0,l}}\bra{\phi_{0,l}}   \nonumber \\
Q_2 =& \ket{\phi_{0,l+m}}\bra{\phi_{0,l+m}}\otimes \ldots \nonumber \\
& \ldots \otimes  \ket{\phi_{0, n-1}}\bra{\phi_{0, n-1}} \otimes \ket{\phi_{0,0}}\bra{\phi_{0,0}}
\end{align}
The operator $Q_1$ is supported on $A$ and $Q_2$ is supported on $C$. Since $A$ and $C$ each contain blocks $M_i$ for at least $l$ values of $i$, conditions (1), (2), and (3) above imply that $Q_1\kettilde{\psi} = Q_2\kettilde{\psi} = \ket{\phi_{0}}/\sqrt{l}$. Thus
\begin{align}
    \Cor(A:C)_{\kettilde{\psi}} \geq& \Tr((Q_1 \otimes Q_2)(\tilde{\rho}_{AC}-\tilde{\rho}_A \otimes \tilde{\rho}_C)) \nonumber\\
    =& \Tr(Q_1 \otimes Q_2 \kettilde{\psi}\bratilde{\psi}) \nonumber \\
    &- \Tr(Q_1 \kettilde{\psi}\bratilde{\psi})\Tr(Q_2 \kettilde{\psi}\bratilde{\psi}) \nonumber \\
    =&1/l-1/l^2.
\end{align}

The choice of $l$ is independent of the chain length $N$, so the above quantity is independent of $N$ and independent of the parameter $m$ measuring the distance between $A$ and $C$. Thus, the correlation certainly does not decay exponentially in the separation between the regions.


\subsection{Proof of Theorem \ref{thm:mainthm}}
In this section, we first state a pair of lemmas that will be essential for the proof of Theorem \ref{thm:mainthm}, then we give a proof summary of Theorem \ref{thm:mainthm}, and finally we provide the full proof of Theorem \ref{thm:mainthm}.

First, an important and well-known tool we use is Uhlmann's theorem \cite{uhlmann1976transition}, which expresses the fact that if two states are close, their purifications will be equally close up to a unitary acting on the purifying auxiliary space. 

\begin{lemma}[Uhlmann's theorem \cite{uhlmann1976transition}]\label{lem:uhlmann}
Suppose $\tau_A$ and $\sigma_A$ are states on system $A$. Suppose $B$ is an auxiliary system and $\ket{T}_{AB}$ and $\ket{S}_{AB}$ are purifications of $\tau_A$ and $\sigma_A$, respectively. Then

\begin{equation}
D(\tau_A,\sigma_A) = \min_{U}\sqrt{1-\lvert\bra{S}_{AB}(I_A \otimes U) \ket T_{AB} \rvert^2},
\end{equation}
where the min is taken over unitaries on system $B$.
\end{lemma}

Second, we prove the following essential statement about states with exponential decay of correlations.

\begin{lemma}\label{lem:tracebound}
If $L$ and $R$ are disjoint regions of a 1D lattice of $N$ sites and the state $\tau = \ket\eta \bra \eta$ has $(t_0, \xi)$-exponential decay of correlations, then
\begin{equation}
D(\tau_{LR}, \tau_L \otimes \tau_R) \leq C_3\exp(-\dist(L,R)/\xi')
\end{equation}
whenever $\dist(L,R) \geq t_0$, where $\xi' = 16\xi^2\log(d)$ and $C_3=\exp(t_0\exp(\tilde{O}(\xi\log(d))))$.

If $\tau$ is the unique ground state of a gapped nearest-neighbor 1D Hamiltonian with spectral gap $\Delta$, then this can be improved to
\begin{equation}
D(\tau_{LR}, \tau_L \otimes \tau_R) \leq C_4\exp(-\dist(L,R)/\xi')
\end{equation}
whenever $\dist(L,R) \geq \Omega(\log^4(d)/\Delta^2)$, where $\xi' = O(1/\Delta)$ and $C_4 = \exp(\tilde{O}(\log^3(d)/\Delta))$.
\end{lemma}

For pure states $\sigma$, we call $\sigma$ a Markov chain for the tripartition $L/M/R$ if $\sigma_{LR} = \sigma_L \otimes \sigma_R$. Thus Lemma \ref{lem:tracebound} states that exponential decay of correlations implies that the violation of the Markov condition, as measured by the purified distance (or alternatively, trace distance) decays exponentially with the size of $M$.

\begin{proof}[Proof of Lemma \ref{lem:tracebound}]

The goal is to show that an exponential decay of correlations in $\tau = \ket{\eta}\bra{\eta}$ implies that $\tau_{LR}$ is close to $\tau_L \otimes \tau_R$. We will do this by truncating the rank of $\tau$ on the region $L$ to form $\sigma$, arguing that $\sigma_{LR}$ is close to $\sigma_L \otimes \sigma_R$, and finally using the triangle inequality to show the same holds for $\tau$.

Lemma \ref{lem:lowrank} says that there is a state $\sigma_L$ with rank $\chi$ defined on region $L$ such that
\begin{equation}
    D(\tau_L,\sigma_L) \leq C_1 e^{-\frac{\log(\chi)}{8 \xi \log(d)}}.
\end{equation}
In fact, the choice of $\sigma_L$ of rank $\chi$ that minimizes the distance to $\tau_L$ is the state $P_L\tau_L/\Tr(P_L\tau_L)$ where $P_L$ is the projector onto the eigenvectors of $\tau_L$ associated with the largest $\chi$ eigenvalues. Accordingly, we define a normalization constant $q=1/\Tr(P_L \tau_L)$ and let $\ket{\nu} = \sqrt{q}P_L\ket{\eta}$ and $\sigma = \ket{\nu}\bra{\nu} = qP_L \tau P_L$ be normalized states.

We first need to show that
\begin{equation}
    \Cor(L:R)_{\ket{\nu}} := \max_{\lVert A \rVert, \lVert B \rVert \leq 1}\Tr((A \otimes B) (\sigma_{LR}-\sigma_L \otimes \sigma_R))
\end{equation}
is small, given only that $\Cor(L:R)_{\ket{\eta}}$ is small. Suppressing tensor product symbols, we can write
\begin{align}
    &{}\Tr((A B) (\sigma_{LR}-\sigma_L \sigma_R)) \nonumber\\
    &= \bra\nu AB \ket \nu - \bra\nu A \ket\nu\bra\nu B \ket\nu \nonumber  \\
    &= q\bra\eta P_LABP_L \ket \eta - q^2\bra\eta P_LAP_L \ket\eta\bra\eta P_L B P_L\ket\eta \nonumber \\
    &= q\bra\eta (P_LAP_L)B \ket \eta - q^2\bra\eta P_LAP_L \ket\eta\bra\eta P_L B \ket\eta \nonumber \\
    &= q\bra\eta (P_LAP_L)B\ket\eta \nonumber \\ 
    &\hspace{6pt}-q^2\bra\eta P_LAP_L \ket\eta\left(\bra\eta P_L B \ket\eta-\bra\eta P_L\ket\eta \bra\eta B \ket\eta\right) \nonumber\\
    &\hspace{6pt}- q^2\bra\eta P_LAP_L \ket\eta\ \bra\eta P_L\ket\eta \bra\eta B \ket\eta \nonumber\\
    &=q\left(\bra\eta (P_LAP_L)B\ket\eta - \bra\eta P_LAP_L\ket\eta \bra\eta B \ket\eta\right) \nonumber \\ 
    &\hspace{6pt}-q^2\bra\eta P_LAP_L \ket\eta\left(\bra\eta P_L B \ket\eta-\bra\eta P_L\ket\eta \bra\eta B \ket\eta\right), 
\end{align}
from which we can conclude
\begin{equation}
    \Cor(L:R)_{\ket{\nu}} \leq (q+q^2) \Cor(L:R)_{\ket{\eta}}.
\end{equation}
The normalization constant $q$ is $1/(1-D(\tau_L,\sigma_L)^2)$ which will be close to 1 as long as $\chi$ is sufficiently large. If we choose $\log(\chi) = 8 \xi \log(d)(1+\log(C_1))$ or larger, then $q$ will certainly be smaller than 2 and $q+q^2 \leq 6$.

The combination of the fact that $\sigma_L$ has small rank and that $\sigma$ has small correlations between $L$ and $R$ will allow us to show that $\sigma_{LR}$ is close to $\sigma_L \otimes \sigma_R$. We do this by invoking Lemma 20 of \cite{brandao2015exponential}, although we reproduce the argument below. We can express the trace norm as

\begin{align}
&\lVert \sigma_{LR}-\sigma_L \otimes \sigma_R \rVert_1 \nonumber\\
={}& \max_{\lVert T \rVert \leq 1} \Tr(T(\sigma_{LR}-\sigma_L \otimes \sigma_R)) \nonumber\\
={}&
\max_{\lVert T \rVert \leq 1} \Tr(P_LTP_L(\sigma_{LR}-\sigma_L \otimes \sigma_R)) ,
\end{align}
where the second equality follows from the fact that $P_L$ fixes the state $\sigma$. We can perform a Schmidt decomposition of the operator $P_LTP_L$ into a sum of at most $\chi^2$ terms which are each a product operator across the $L/R$ partition
\begin{equation}
    P_LTP_L = \sum_{j=1}^{\chi^2} T_{L,j} \otimes T_{R,j}
\end{equation}
and also such that $\lVert T_{L,j} \rVert, \lVert T_{R,j} \rVert \leq 1$ (see Lemma 20 of \cite{brandao2015exponential} for full justification of this). Then we may write
\begin{align}
&\lVert \sigma_{LR}-\sigma_L \otimes \sigma_R \rVert_1 \nonumber\\
\leq{}& \max_{\lVert T \rVert \leq 1} \Tr\left(\left(\sum_{j=1}^{\chi^2}T_{L,j} \otimes T_{R,j}\right)(\sigma_{LR}-\sigma_L \otimes \sigma_R)\right) \nonumber \\
\leq{}& \sum_{j=1}^{\chi^2}\max_{\lVert T_{L,j} \rVert,\lVert T_{R,j} \rVert \leq 1} \Tr\left(\left(T_{L,j} \otimes T_{R,j}\right)(\sigma_{LR}-\sigma_L \otimes \sigma_R)\right) \nonumber \\
\leq{}& \chi^2 \Cor(L:R)_{\ket{\nu}} \nonumber \\
\leq{}& 6\chi^2 \Cor(L:R)_{\ket{\eta}} \leq 6 \chi^2 \exp(-\dist(L,R)/\xi)
\end{align}
as long as $\chi \geq 8 \xi \log(d)(1+\log(C))$ and $\dist(L,R) \geq t_0$.

Moreover the purified distance is bounded by the square root of the trace norm of the difference (Lemma \ref{lem:purifiedvstrace}), allowing us to say
\begin{equation}
D(\sigma_{LR}, \sigma_L \otimes \sigma_R) \leq \sqrt{6}\chi \exp(-\dist(L,R)/(2\xi)).
\end{equation}

Then, by the triangle inequality, we can bound
\begin{align}
&D(\tau_{LR},\tau_L \otimes \tau_R) \nonumber\\
\leq{}& D(\tau_{LR}, \sigma_{LR}) + D(\sigma_{LR},\sigma_L \otimes \sigma_R) \nonumber \\
&+D(\sigma_L \otimes \sigma_R, \tau_L \otimes \tau_R) \nonumber \\
\leq{}& D(\tau_{LR}, \sigma_{LR}) + D(\sigma_{LR},\sigma_L \otimes \sigma_R) \nonumber \\
&+D(\sigma_L, \tau_L) + D( \sigma_R , \tau_R) \nonumber\\
\leq{}& 3C_1\exp\left(-\frac{\log(\chi)}{8\xi\log(d)}\right)+\sqrt{6}\chi \exp\left(-\frac{\dist(L,R)}{2\xi}\right).
\end{align}
We can choose 
\begin{equation}
    \log(\chi) = 8\xi\log(d)(1+\log(3C_1)) + \dist(L,R)/(4\xi).
\end{equation}
Then each term can be bounded so that
\begin{equation}
    D(\tau_{LR},\tau_L \otimes \tau_R) \leq 2\sqrt{6}(3eC_1)^{8\xi\log(d)}e^{-\frac{\dist(L,R)}{32\xi^2\log(d)}},
\end{equation}
which proves the first part of the Lemma.

If $\tau$ is the unique ground state of a gapped Hamiltonian, then we may use the second part of Lemma \ref{lem:lowrank}, and bound

\begin{align}
&D(\tau_{LR},\tau_L \otimes \tau_R) \nonumber\\
\leq{}& 3C_2 e^{-\tilde{O}\left(\frac{\Delta^{1/3}\log^{4/3}(\chi)}{\log^{4/3}(d)}\right)}+\sqrt{6}\chi e^{-O\left(\Delta\dist(L,R)\right)}.
\end{align}
Here we can choose
\begin{equation}
    \log(\chi) = O\left(\frac{\log(d)(1+\log(3C_2))^{\frac{3}{4}}}{\Delta^{1/4}} + \Delta\dist(L,R)\right),
\end{equation}
and then each term is small enough to make the bound
\begin{align}
    &D(\tau_{LR},\tau_L \otimes \tau_R)\nonumber\\
    \leq{}& e^{\tilde{O}\left(\frac{\Delta^{5/3}\dist(L,R)^{4/3}}{\log^{4/3}(d)}\right)}\nonumber\\
   & +\sqrt{6}e^{O\left(\frac{\log(d)\log^{3/4}(3eC_2)}{\Delta^{1/4}}\right)}e^{-O\left(\Delta\dist(L,R)\right)} \nonumber \\
    \leq{}& e^{\tilde{O}\left(\frac{\log^3(d)}{\Delta}\right)}e^{-O\left(\Delta\dist(L,R)\right)}
\end{align}
as long as $\dist(L,R) \geq \Omega(\log^4(d)/\Delta^2)$, so that the first term in the second line is dominated by the second term. Also note we have used $\log(C_2) = \tilde{O}(\log^{8/3}(d)/\Delta)$ in the last line. This proves the second part of the lemma.
\end{proof}


\begin{figure}[ht]
    \centering
    \includegraphics[width = \columnwidth]{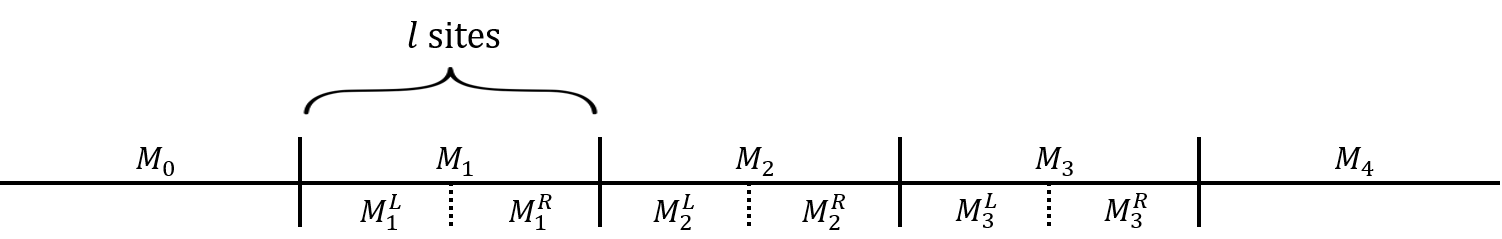}
    \caption{Schematic for the proof of Theorem \ref{thm:mainthm}. The chain is divided into regions $M_i$ of length $l$, which are themselves divided into left and right halves $M_i^L$ and $M_i^R$. The state $\kettilde{\psi}$ is constructed by starting with $\ket{\psi}$, applying unitaries that act only on $M_i$ for each $i$ to disentangle the state across the $M_i^L/M_i^R$ cut, projecting onto a product state across those cuts, and finally applying the inverse unitaries on regions $M_i$.}
    \label{fig:mainthm}
\end{figure}

\begin{proof}[Proof summary for Theorem \ref{thm:mainthm}]

First, we make the following observation about tripartitions of the chain into contiguous regions $L$, $M$, and $R$: since $\ket{\psi}$ has exponential decay of correlations, the quantity $D(\rho_{LR}, \rho_L \otimes \rho_R)$ is exponentially small in $\lvert M \rvert / \xi'$, where $\xi' = O(\xi^2 \log(d))$. This is captured in Lemma \ref{lem:tracebound} and requires the area law result from \cite{brandao2015exponential}. One can truncate $\rho_L$ and $\rho_R$ to bond dimension $d^{\lvert M \rvert/2}$, incurring small error, then purify $\rho_L$ into $\ket\alpha$ using the left half of region $M$ as the purifying auxiliary space and $\rho_R$ into $\ket\beta$ using the right half. Since $\ket{\alpha} \otimes \ket{\beta}$ and $\ket{\psi}$ are nearly the same state after tracing out $M$, there is a unitary $U$ acting only on $M$ that nearly disentangles $\ket\psi$ across the central cut of $M$, with $U\ket{\psi}\approx \ket{\alpha}\otimes\ket{\beta}$ (Uhlmann's theorem, Lemma \ref{lem:uhlmann}).

The proof constructs the approximation $\kettilde{\psi}$ by applying three steps of operations on the exact state $\ket{\psi}$. First, the chain is broken up into many regions $\{M_i\}_{i=0}^{n+1}$ of length $l$, and  disentangling unitaries $U_i$ as described above are applied to each region $M_i$ in parallel. The state is close to, but not exactly, a product state across the center cut of each region $M_i$. To make it an exact product state, the second step is to apply rank-1 projectors $\Pi_i$ onto the right half of the state across each of these cuts, starting with the leftmost cut and working our way down the chain. Then, the third step is to apply the reverse unitaries $U_i^\dagger$ that we applied in step 1. 

The projection step is the cause of the error between $\ket{\psi}$ and $\kettilde{\psi}$. The number of projections is $O(N)$, but the error accrued locally is only constant. This follows from the fact that the projectors are rank 1, so once we apply projector $\Pi_i$, region $M_j$ is completely decoupled from $M_{j'}$ when $j<i$ and $j' > i$. Thus any additional projections, which act only on the regions to the right of $M_i$, have no effect on the reduced density matrix on $M_j$ (except to reduce its norm). Using this logic, we show that the number of errors that actually affect the state locally on a region $X$ is proportional to the number of sites in $X$, and not the number of sites in the whole chain. To make this error less than $\epsilon$ for any region of at most $k$ sites, we can choose $l = O(\xi'\log(k/\epsilon))$.

After step 2, the state is a product state on blocks of $l$ sites each, and in step 3, unitaries are applied that couple neighboring blocks, so the maximum Schmidt rank across any cut cannot exceed $d^l$. This yields the scaling for $\chi$.

The result is improved when $\ket{\psi}$ is the ground state of a gapped local Hamiltonian by using the improved area law \cite{arad2013area,arad2017rigorous} in the proof of Lemma \ref{lem:tracebound} and in the truncation of the states $\rho_L$ and $\rho_R$ before purifying into $\ket{\alpha}$ and $\ket{\beta}$.

Finally, it can be seen that $\kettilde{\psi}$ is formed by a constant-depth quantum circuit with two layers of unitaries, where each unitary acts on $l$ qubits. The first layer prepares the product state over blocks of length $l$ that is attained after applying projections in step 2, and the second layer applies the inverse unitaries from step 3. Each unitary in this circuit can be decomposed into a sequence of nearest-neighbor gates with depth $\tilde{O}(d^{2l})$. 
\end{proof}

\begin{proof}[Proof of Theorem \ref{thm:mainthm}]

We fix an even integer $l$, which we will specify later, and divide the $N$ sites of the chain into $n+2$ segments of length $l$, which we label, from left to right: $M_0,M_1,\ldots, M_{n+1}$. If $N$ does not divide $l$ evenly, then we allow segment $M_{n+1}$ to have fewer than $l$ sites. For $i \in [1,n]$ let $L_i$ be the sites to the left of region $M_i$ and let $R_i$  be the sites to the right of region $M_i$.

Lemma \ref{lem:tracebound} tells us that, since $\ket \psi$ has $(t_0,\xi)$-exponential decay of correlations, for any $i \in [1,n]$, $\rho_{L_iR_i}$ is close to $\rho_{L_i} \otimes \rho_{R_i}$ when $l$ is much larger than $\xi'$; that is,
\begin{equation}
D(\rho_{L_iR_i},\rho_{L_i} \otimes \rho_{R_i}) \leq C_3\exp(-l/\xi')
\end{equation}
whenever $l \geq t_0$. We also choose $\chi = d^{l}$ and for each $i$ define $\rho'_{L_i}$ and $\rho'_{R_i}$, each with rank at most $\sqrt\chi$ by taking $A = L_i$ and $A = R_i$ in Lemma \ref{lem:lowrank}. Thus we have
\begin{align}
D(\rho_{L_i}, \rho'_{L_i}) &\leq C_1 \exp(-l/(16 \xi)), \\
D(\rho_{R_i}, \rho'_{R_i}) &\leq C_1 \exp(-l/(16 \xi)).
\end{align}
Then by the triangle inequality we have
\begin{align}
&D(\rho_{L_iR_i},\rho'_{L_i} \otimes \rho'_{R_i}) \nonumber\\
\leq{}& D(\rho_{L_iR_i},\rho_{L_i} \otimes \rho_{R_i}) + D(\rho_{L_i} \otimes \rho_{R_i},\rho'_{L_i} \otimes \rho_{R_i}) \nonumber\\
& + D(\rho'_{L_i} \otimes \rho_{R_i},\rho'_{L_i} \otimes \rho'_{R_i}) \nonumber\\
={}& D(\rho_{L_iR_i},\rho_{L_i} \otimes \rho_{R_i}) + D(\rho_{L_i},\rho'_{L_i}) + D(\rho_{R_i},\rho'_{R_i}) \nonumber\\
\leq{}& C \exp(-l/\xi''), 
\end{align}
where $C\leq 2C_1+C_3$, $\xi'' = \max(\xi', 16 \xi)$, whenever $l \geq t_0$. 

Note that $\ket{\psi}_{L_iM_iR_i}$ can be viewed as a purification of $\rho_{L_iR_i}$ with $M_i$ the purifying auxiliary system. Divide region $M_i$ in half, forming $M_i^L$ (left half) and $M_i^R$ (right half). See Figure \ref{fig:mainthm} for a schematic. Each of these subsystems has total dimension $d^{l/2}$ and thus can act as the purifying auxiliary system for $\rho'_{L_i}$ or $\rho'_{R_i}$. Let $\ket{\alpha_i}_{L_iM_i^L}$ be a purification of $\rho'_{L_i}$ and $\ket{\beta_i}_{M_i^RR_i}$ be a purification of $\rho'_{R_i}$. Thus, $\ket{\alpha_i} \otimes \ket{\beta_i}$, which is defined over the entire original chain, is a purification of $\rho'_{L_i} \otimes \rho'_{R_i}$.

Uhlmann's theorem (Lemma \ref{lem:uhlmann}) shows how these purifications are related by a unitary on the purifying auxiliary system: for each $M_i$ with $i \in [1,n]$, there is a unitary $U_i$ acting non-trivially on region $M_i$ and as the identity on the rest of the chain such that $U_i\ket{\psi}$ is very close to $\ket{\alpha_i} \otimes \ket{\beta_i}$, a product state across the cut between $M_i^L$ and $M_i^R$. In other words, $U_i$ disentangles $L_i$ from $R_i$, up to some small error, by acting only on $M_i$. Formally we say that 
\begin{equation}
\left\lvert\bra{\alpha_i}_{L_iM_i^L}\otimes \bra{\beta_i}_{M_i^RR_i}U_i  \ket{\psi}_{L_iM_iR_i}\right\rvert = \sqrt{1-\delta_i^2},
\end{equation}
where $\delta_i \leq C\exp(-l/\xi'')$ for all $i$.

An equivalent way to write this fact is 
\begin{align}\label{eq:Vipsi}
U_i\ket{\psi}_{L_iM_iR_i} ={}& \sqrt{1-\delta_i^2}\ket{\alpha_i}_{L_iM_i^L}\otimes \ket{\beta_i}_{M_i^RR_i} \nonumber \\
&+ \delta_i\ket{\phi_i'}_{L_iM_i^LM_i^RR_i}
\end{align}
where $\ket{\phi_i'}$ is a normalized state orthogonal to $\ket{\alpha_i} \otimes \ket{\beta_i}$.

We can define the projector
\begin{equation}
\Pi_i = I_{L_iM_i^L}\otimes \ket{\beta_i}\bra{\beta_i}_{M_i^RR_i},
\end{equation}
whose rank is 1 when considered as an operator acting only on $M_i^R R_i$. 

We notice that $\Pi_iU_i\ket{\psi}_{L_iM_iR_i}$ is a product state across the $M_i^L/M_i^R$ cut and has a norm close to 1. Suppose we alternate between applying disentangling operations $U_i$ and projections $\Pi_i$ onto a product state as we move down the chain. Each $\Pi_i$ will reduce the norm of the state, but we claim the norm will never vanish completely (and delay the proof for later for clarity of argument).

\begin{claim}\label{prop:normnonzero}
If $l \geq \xi'' \log(3C)$, then
\begin{equation}
\lVert \Pi_nU_n \ldots \Pi_1 U_1 \ket{\psi} \rVert \neq 0.
\end{equation}

\end{claim}

This allows us to define 
\begin{equation}
\kettilde{\phi} = \frac{\Pi_nU_n \ldots \Pi_1 U_1 \ket{\psi}}{\lVert \Pi_nU_n \ldots \Pi_1 U_1 \ket{\psi} \rVert}.
\end{equation}

Note that, to put our proof in line with what is described in the introduction and proof summary earlier, we may act with all the unitaries prior to the projectors if we conjugate the projectors
\begin{equation}
\kettilde{\phi} \propto \Pi'_n\ldots \Pi'_1 U_n \ldots U_1 \ket{\psi},
\end{equation}
where $\Pi'_i = U_n\ldots U_{i+1}\Pi_iU_{i+1}^\dagger\ldots U_{n}^\dagger$, which still only acts on the region $R_i$. 

This can be compared with the state $\ket{\phi}$ defined by applying the disentangling operations without projecting:
\begin{equation}
\ket{\phi} = U_n \ldots U_1 \ket{\psi}.
\end{equation}
We claim that $\kettilde{\phi}$ is a good local approximation for $\ket{\phi}$ (and delay the proof for clarity of argument).

\begin{claim}\label{prop:philocalapprox}
For any integer $k'$, $\kettilde{\phi}$ is a $(k', \epsilon')$-local approximation to $\ket{\phi}$ with $\epsilon'=C\sqrt{k'/l+3}\exp(-l/\xi'')$.
\end{claim}

Next we can define
\begin{equation}\label{eq:tildepsifromtildephi}
\kettilde{\psi} = U_n^\dag \ldots U_1^\dag \kettilde{\phi},
\end{equation}
which parallels the relationship
\begin{equation}
\ket{\psi} = U_n^\dag \ldots U_1^\dag \ket{\phi}.
\end{equation}

Now suppose $X$ is a contiguous region of the chain of length $k$. Then there is a region $X'$ of the chain of length at most $k'=k+2l$ that contains $X$ and is made up of regions $M_j$ where $j \in [a',b']$. Then

\begin{align}\label{eq:psitildelocalapprox}
&\lVert \Tr_{X^c}(\kettilde{\psi}\bratilde{\psi}-\ket{\psi}\bra{\psi})\rVert_1 \nonumber\\
\leq{}& \lVert \Tr_{X'^c}(\kettilde{\psi}\bratilde{\psi}-\ket{\psi}\bra{\psi})\rVert_1 \nonumber\\
={}&\lVert \Tr_{X'^c}(\kettilde{\phi}\bratilde{\phi}-\ket{\phi}\bra{\phi})\rVert_1 \nonumber \\
\leq{}& C\sqrt{k/l+5}\exp(-l/\xi'') \nonumber \\
\leq{}& C\sqrt{6k}\exp(-l/\xi''),
\end{align}
where the third line follows from the fact that $\ket{\phi}$ and $\ket{\psi}$ are related by a unitary that does not couple region $X'$ and region $X'^c$, and the fourth line follows from Claim \ref{prop:philocalapprox}. 

If we choose $l = \max(t_0,\xi''\log(3C\sqrt{k}/\epsilon))$, then the requirements of Claim \ref{prop:normnonzero} are satisfied, and we can see from Eq.~\eqref{eq:psitildelocalapprox} that $\kettilde{\psi}$ is a $(k,\epsilon)$-local approximation to $\ket{\psi}$, item (1) of the theorem.

Item (2) states that $\kettilde{\psi}$ can be written as an MPS with constant bond dimension. This can be seen by the following logic. The Schmidt rank of $\kettilde{\phi}$ across any cut $M_i^L/M_i^R$ is 1, as discussed in the proof of Claim \ref{prop:philocalapprox} (see Eq.~\eqref{eq:phitildeproduct}), since the projector $\Pi_i$ projects onto a product state across that cut and unitaries $U_j$ with $j>i$ act trivially across the cut. By Corollary \ref{cor:schmidtrankincrease}, this implies that the Schmidt rank across any cut $M_i^R/M_{i+1}^L$ can be at most $d^{l/2}$. Acting with the inverse unitaries $U_j^\dag$ on $\kettilde{\phi}$ to form $\kettilde{\psi}$ preserves the Schmidt rank across the cut $M_i^R/M_{i+1}^L$, since none couple both sides of the cut. Because any cut is at most distance $l/2$ from some $M_i^R/M_{i+1}^L$ cut, the Schmidt rank across an arbitrary cut can be at most a factor of $d^{l/2}$ greater, again by Corollary \ref{cor:schmidtrankincrease}, meaning the maximum Schmidt rank across any cut of $\ket{\psi}$ is $\chi=d^l$. 

Given our choice of $l=\xi''\log(C\sqrt{6k}/\epsilon)$, we find that the state can be represented by an MPS with bond dimension 
\begin{align}
\chi &= (\sqrt{6}C)^{\xi''\log(d)}(k/\epsilon^2)^{\xi''\log(d)/2} \nonumber\\
&= e^{e^{\tilde{O}(\xi \log(d))}}(k/\epsilon^2)^{O(\xi^2\log^2(d))}.
\end{align}
This proves item (2). Note that, in the case that our choice of $l$ exceeds $N$, it is not possible to form the construction we have described. However, in this case $d^l$ will exceed $d^{N}$ and we may take $\kettilde{\psi} =\ket{\psi}$, which is a local approximation for any $k$ and $\epsilon$ and has bond dimension in line with item (2) or (2').

Item (2') follows by using the same equations with $C=2C_2+C_4 = \exp(\tilde{O}(\log^{3}(d)/\Delta))$ and $\xi'' =O(1/\Delta)$. For Lemma \ref{lem:tracebound} to apply, we must have $l \geq \Omega(\log^4(d)/\Delta^2)$, but this will be satisfied for sufficiently large choices of $k/\epsilon^2$. Thus the final analysis yields
\begin{equation}
\chi = e^{\tilde{O}(\log^{4}(d)/\Delta^{2})}(k/\epsilon^2)^{O\left(\log(d)/\Delta\right)}.
\end{equation}

Item (3) states that $\kettilde{\psi}$ can be formed from a low-depth quantum circuit. In the proof of Claim 2, we show how the state $\kettilde{\phi}$ is a product state across divisions $M_i^L/M_{i}^R$, as in Eq.~\eqref{eq:phitildeproduct}. Thus the state $\kettilde{\phi}$ can be created from $\ket{0}^{\otimes N}$ by acting with non-overlapping unitaries on regions $M_i^RM_{i+1}^L$ in parallel. Each of these unitaries is supported on $l$ sites. Then, $\kettilde{\psi}$ is related to $\kettilde{\phi}$ by another set of non-overlapping unitaries supported on $l$ sites, as shown in Eq.~\eqref{eq:tildepsifromtildephi}. We conclude that $\kettilde{\psi}$ can be created from the trivial state by two layers of parallel unitary operations where each unitary is supported on $l$ sites, as illustrated in Figure \ref{fig:constantdepthcircuit}. In \cite{brennen2005efficient}, it is shown how any $l$-qudit unitary can be decomposed into $O(d^{2l}) = O(\chi^2)$ two-qudit gates, with no need for ancillas. We can guarantee that these gates are all spatially local by spending at most depth $O(l)$ performing swap operations to move any two sites next to each other, a factor only logarithmic in the total depth. This proves the theorem.
\end{proof}


\begin{proof}[Proof of Claim \ref{prop:normnonzero}]

We prove by induction. Let $|\tilde{\phi}_j\rangle = \Pi_jU_j \ldots \Pi_1U_1 \ket{\psi}$. Note that $\Pi_i U_i\ket{\psi}$ is non-zero for all $i$, so in particular $|\tilde{\phi}_1\rangle$ is non-zero. Furthermore we note that, if it is non-zero, $|\tilde{\phi}_j\rangle$ can be written as a product state $\ket{\alpha'_j}_{L_jM_j^L}\otimes \ket{\beta_j}_{M_j^RR_j}$ for some unnormalized but non-zero state $\ket{\alpha'_j}$, and the reduced density matrix of $|\tilde{\phi}_j\rangle$ on $R_j$ is $\rho'_{R_j}$. If we assume $|\tilde{\phi}_j\rangle$ is non-zero then we can write
\begin{align}
|\tilde{\phi}_{j+1}\rangle &= \Pi_{j+1}U_{j+1} |\tilde{\phi}_j\rangle \nonumber \\
&= \ket{\alpha_j'}\otimes \Pi_{j+1}U_{j+1} \ket{\beta_j}
\end{align}
and
\begin{align}
\lVert |\tilde{\phi}_{j+1}\rangle\rVert^2 
&= \lVert \ket{\alpha_j'}\otimes \Pi_{j+1}U_{j+1} \ket{\beta_j}\rVert^2 \nonumber\\ 
&= \lVert \ket{\alpha'_j}\rVert^2\lVert \Pi_{j+1}U_{j+1} \ket{\beta_j}\rVert^2 \nonumber\\
&= \lVert \ket{\alpha'_j}\rVert^2\Tr(\Pi_{j+1}U_{j+1}\ket{\beta_j}\bra{\beta_j}U_{j+1}^\dag \Pi_{j+1}) \nonumber\\
&= \lVert \ket{\alpha'_j}\rVert^2\Tr(\Pi_{j+1}U_{j+1}\rho'_{R_{j}}U_{j+1}^\dag \Pi_{j+1})\nonumber \\
&\geq \lVert \ket{\alpha'_j}\rVert^2\Tr(\Pi_{j+1}U_{j+1}\rho_{R_j} U_{j+1}^\dag \Pi_{j+1}) \nonumber\\
&-\lVert \ket{\alpha'_j}\rVert^2\lVert \rho_{R_j}- \rho'_{R_j}\rVert_1 \nonumber\\
&\geq \lVert \ket{\alpha'_j}\rVert^2\Tr(\Pi_{j+1}U_{j+1}\rho U_{j+1}^\dag \Pi_{j+1}) \nonumber\\
&-\lVert \ket{\alpha'_j}\rVert^2\lVert \rho_{R_j}- \rho'_{R_j}\rVert_1 \nonumber\\
&\geq \lVert \ket{\alpha'_j}\rVert^2\lVert\Pi_{j+1}U_{j+1}\ket{\psi} \rVert^2 \nonumber\\
&-\lVert \ket{\alpha'_j}\rVert^2\lVert \rho_{R_j}- \rho'_{R_j}\rVert_1 \nonumber\\
&\geq \lVert \ket{\alpha'_j}\rVert^2(1-C\exp(-l/\xi''))^2 \nonumber\\
&- \lVert \ket{\alpha'_j}\rVert^2C\exp(-l/\xi'') \nonumber\\
& >  0 \nonumber
\end{align}
as long as $l \geq \xi'' \log(3C)$.
\end{proof}


\begin{proof}[Proof of Claim \ref{prop:philocalapprox}]

First, consider the cut $M_i^L/M_i^R$ during the formation of the state $\kettilde{\phi}$. When the projector $\Pi_i$ is applied, the state becomes a product state across this cut. The remaining operators are $U_j$ and $\Pi_j$ with $j>i$, and thus they have no effect on the Schmidt rank across the $M_i^L/M_i^R$ cut, meaning $\kettilde{\phi}$ is a product state across each of these cuts, or in other words
\begin{align}\label{eq:phitildeproduct}
\kettilde{\phi} ={}& \ket{\phi_1}_{M_0M_1^L}\otimes \ket{\phi_2}_{M_1^RM_2^L} \otimes \ldots \nonumber \\
&\ldots \otimes \ket{\phi_n}_{M_{n-1}^RM_n^L} \otimes \ket{\phi_{n+1}}_{M_n^RM_{n+1}}.
\end{align}

Given an integer $k$ and a contiguous region $X$ of length $k$, we can find integers $a$ and $b$ such that $Y = M_a^RM_{a+1}\ldots M_{b-1}M_b^L$ contains $X$ and $\lvert b - a \rvert \leq k/l+2$. Then
\begin{align}
&\Tr_{Y^c}(\kettilde{\phi}\bratilde{\phi}) \nonumber \\
={}& \ket{\phi_{a+1}}\bra{\phi_{a+1}} \otimes \ldots \otimes \ket{\phi_{b}}\bra{\phi_{b}} \nonumber\\
\propto{}& \Tr_{L_aM_a^LM_b^RR_b}\left(\Pi_bU_b\ldots\Pi_aU_a\ket{\psi}\bra{\psi}U_a^\dag\Pi_a\ldots U_b^\dag \Pi_b\right).
\end{align}
The advantage here is that all of the $U_i$ and $\Pi_i$ for which $i \not\in [a,b]$ have disappeared.
On the other hand, we have
\begin{align}
&\Tr_{Y^c}(\ket{\phi}\bra{\phi}) \nonumber\\
={}& \Tr_{Y^c}(U_n\ldots U_1\ket{\psi}\bra{\psi}U_1^\dag \ldots U_n^\dag) \nonumber\\
={}& \Tr_{L_aM_a^LM_b^RR_b}\left(U_b\ldots U_{a}\ket{\psi}\bra{\psi}U_{a}^\dag\ldots U_{b}^\dag\right).
\end{align}
Note that, since
\begin{equation}
U_i \ket{\psi} = \sqrt{1-\delta_i^2}\ket{\alpha_i} \otimes \ket{\beta_i} + \delta_i \ket{\phi_i'},
\end{equation}
we can say that
\begin{equation}
\Pi_i U_i \ket{\psi} = U_i\ket{\psi} - \delta_i(I-\Pi_i)\ket{\phi_i'},
\end{equation}
and thus
\begin{align}
&\Pi_bU_b\ldots \Pi_aU_a\ket{\psi} \nonumber \\
={}& U_b\ldots U_a \ket{\psi} - \nonumber\\
&\sum_{j=a}^b \delta_j(\Pi_bU_b\ldots \Pi_{j+1}U_{j+1})(I-\Pi_j)\ket{\phi_j'} \nonumber\\
\equiv{}& U_b\ldots U_a \ket{\psi} - \delta \ket{\phi'},
\end{align}
where $\delta \leq \sqrt{\sum_{j=a}^b\delta_j^2}$ and normalized $\ket{\phi'}$ is normalized. This implies
\begin{equation}
    \frac{\lvert\bra{\psi} U_a^\dagger \ldots U_b^\dagger \Pi_b U_b \ldots \Pi_a U_a \ket{\psi}\rvert}{\lVert \Pi_a \ldots U_b^\dagger \Pi_b U_b \ldots \Pi_a U_a \ket{\psi}\rVert} \geq \sqrt{1-\delta^2},
\end{equation}
which shows that $D_1(\tau, \tau') \leq \delta$, where 
\begin{align}
&\tau = U_b\ldots U_{a}\ket{\psi}\bra{\psi}U_{a}^\dag\ldots U_{b}^\dag \nonumber\\
&\tau' = \frac{\Pi_bU_b\ldots\Pi_aU_a\ket{\psi}\bra{\psi}U_a^\dag\Pi_a\ldots U_b^\dag \Pi_b}{\lVert \Pi_a \ldots U_b^\dagger \Pi_b U_b \ldots \Pi_a U_a \ket{\psi}\rVert^2} 
\end{align}
and hence
\begin{align}
&D_1( \Tr_{X^c}(\ket{\phi}\bra{\phi}),\Tr_{X^c}(\kettilde{\phi}\bratilde{\phi}))\nonumber \\
&\leq D_1( \Tr_{Y^c}(\ket{\phi}\bra{\phi}),\Tr_{Y^c}(\kettilde{\phi}\bratilde{\phi}))\nonumber \\
&\leq D_1(\tau_Y,\tau'_Y) \leq D_1(\tau,\tau') \leq \delta \nonumber \\
&\leq C\sqrt{k/l+3}\exp(-l/\xi'').
\end{align}
This holds for any region $X$ of length $k$, so this proves the claim.
\end{proof}


\subsection{Proof of Theorem \ref{thm:reduction}}

First we state a lemma that will do most of the legwork needed for Theorem \ref{thm:reduction}. Then we prove Theorem \ref{thm:reduction}.

\begin{lemma}\label{lem:Kcombined}
Suppose, for $j=0,1$, $H^{(j)}$ is a translationally invariant Hamiltonian defined on a chain of length $N$ and local dimension $d$. Further suppose that $|\psi^{(j)}\rangle$ is the unique ground state of $H^{(j)}$ with energy $E_0^{(j)}$, and let $\Delta^{(j)}$ be the spectral gap of $H^{(j)}$. We may form a new chain with local dimension $2d$ by adding an ancilla qubit to each site of the chain. Then there is a Hamiltonian $K$ defined on this chain such that
\begin{enumerate}[(1)]
    \item The ground state energy of $K$ is \begin{equation}
        E_0^K = \frac{1}{3}\min_jE_0^{(j)}.
    \end{equation}
    \item If $E_0^{(0)} < E_0^{(1)}$, then the ground state of $K$ is $\ket{0^N}_A \otimes |\psi^{(0)}\rangle$ and if $E_0^{(1)} < E_0^{(0)}$, then the ground state is $\ket{1^N}_A \otimes |\psi^{(1)}\rangle$, where $A$ refers to the $N$ ancilla registers collectively.
    \item If $E_0^{(0)} < E_0^{(1)}$, then the spectral gap of $K$ is at least $\min(\Delta^{(0)},E_0^{(1)} - E_0^{(0)},1)/3$ and if $E_0^{(1)} < E_0^{(0)}$, then the spectral gap of $K$ is at least $\min(\Delta^{(1)},E_0^{(0)} - E_0^{(1)},1)/3$.
\end{enumerate}
\end{lemma}

\begin{proof}[Proof of Lemma \ref{lem:Kcombined}]
Note that a variant of this lemma is employed in \cite{bausch2018undecidability,cubitt2015undecidability,bausch2018size} to show the undecidability of certain properties of translationally invariant Hamiltonians.

Since $H^{(j)}$ is translationally invariant, it is specified by its single interaction term $H^{(j)}_{i,i+1}$:
\begin{equation}
    H^{(j)} = \sum_{i=1}^{N-1} H^{(j)}_{i,i+1}.
\end{equation}
The Hamiltonian $K$ will be defined over a new chain where we attach to each site an ancilla qubit, increasing the local Hilbert space dimension by a factor of 2. We refer to the ancilla associated with site $i$ by the subscript $A_i$, and we refer to the collection of ancillas together with the subscript $A$. Operators or states without a subscript are assumed to act on the original $d$-dimensional part of the local Hilbert spaces. Let
\begin{equation}
    K =  \sum_{i=1}^{N-1} K_{i,i+1},  
\end{equation}
where
\begin{align}
    K_{i,i+1} ={}& \;\;\; \frac{1}{3}H^{(0)}_{i,i+1} \otimes \ket{00}\bra{00}_{A_iA_{i+1}} \\
    &+ \frac{1}{3}H^{(1)}_{i,i+1} \otimes \ket{11}\bra{11}_{A_iA_{i+1}} \\
    &+ I_{i,i+1} \otimes\left(\ket{01}\bra{01}+\ket{10}\bra{10}\right)_{A_iA_{i+1}}
\end{align}
with $I_{i,i+1}$ denoting the identity operation on sites $i$ and $i+1$. In this form it is clear that $K$ is a nearest-neighbor translationally invariant Hamiltonian and that each interaction term has operator norm 1 (a requirement under our treatment of 1D Hamiltonians). The picture we get is that if two neighboring ancillas are both $\ket 0$, then $H^{(0)}_{i,i+1}/3$ is applied, if both are $\ket 1$ then $H^{(1)}_{i,i+1}/3$ is applied, and if the ancillas are different, then $I_{i,i+1}$ is applied. Following this intuition, we can rewrite $K$ as follows.
\begin{equation}
    K = \sum_{x=0}^{2^N-1} \left(\ket{x} \bra{x}_A\otimes \sum_{i=1}^{N-1} K_{x,i}\right),
\end{equation}
where the first sum is over all settings of the ancillas and the operator $K_{x,i}$ acts on the non-ancilla portion of sites $i$ and $i+1$, with 
\begin{equation}
    K_{x,i}=
\begin{cases} 
      H^{(0)}_{i,i+1}/3 & \text{ if } x_i=x_{i+1}=0 \\
      H^{(1)}_{i,i+1}/3 & \text{ if }x_i=x_{i+1}=1 \\
      I_{i,i+1} & \text{ if }x_i \neq x_{i+1} 
   \end{cases}
\end{equation}

We analyze the spectrum of $K$. If $H^{(j)}$ has eigenvalues $E^{(j)}_n$ with corresponding eigenvectors $|\phi^{(j)}_n\rangle$ (where $E^{(j)}_n$ is non-decreasing with increasing integers $n$), then the states $\ket{0^N}_A \otimes |\phi^{(0)}_n\rangle$ and $\ket{ 1^N}_A \otimes |\phi^{(1)}_n\rangle$ are eigenstates of $K$ with eigenvalues $E^{(0)}_n/3$ and $E^{(1)}_n/3$, respectively.

Recall that eigenvectors of a Hamiltonian span the whole Hilbert space over which the Hamiltonian is defined. Therefore, the eigenvectors of $K$ listed above span the entire sectors of the Hilbert space associated with the ancillas set to $\ket{0^N}_A$ or $\ket{ 1^N}_A$. Suppose $\ket{\phi}$ is another eigenvector of $K$. Since it is orthogonal to all of the previously listed eigenvectors, $\ket{\phi}$ can be written
\begin{equation}
    \ket{\phi} = \sum_{x=1}^{2^N-2} \alpha_x \ket x_A \otimes \ket{\eta_x}
\end{equation}
for some set of complex coefficients $\alpha_x$ with $\sum_x \lvert \alpha_x \rvert^2=1$ and some set of normalized states $\ket{\eta_x}$. The sum explicitly leaves out the $x=0 = 0^N$ and $x=2^N-1 = 1^N$ binary strings because these states lie in the subspace spanned by eigenstates already listed. We wish to lower bound the energy of the state $\ket{\phi}$, i.e.,~the quantity
\begin{align}
    &\bra \phi K\ket \phi\nonumber\\
    &= \sum_{x=1}^{2^N-2}\sum_{y=1}^{2^N-2} \alpha^*_x\alpha_y \bra{x}_A \bra{\eta_x}K\ket{y}_A \ket{ \eta_y} \nonumber\\
    &=\sum_{x=1}^{2^N-2}\sum_{y=1}^{2^N-2}\sum_{z=0}^{2^N-1} \alpha^*_x\alpha_y \braket{x}{z}\braket{z}{y}_A \bra{\eta_x}\sum_{i=1}^{N-1} K_{z,i}\ket{ \eta_y} \nonumber\\
    &=\sum_{x=1}^{2^N-2} \lvert \alpha_x\rvert^2 \bra{\eta_x}\sum_{i=1}^{N-1} K_{x,i}\ket{\eta_x}.
\end{align}
We make the following claim:
\begin{claim}\label{claim:partialchain}
For any state $\ket{\eta}$, and any $1 \leq a < b \leq N$
\begin{equation}
     \bra{\eta} \sum_{i=a}^{b-2}H^{(j)}_{i,i+1} \ket{\eta} \geq \frac{b-a}{N-1}E_0^{(j)}-1.
\end{equation}
\end{claim}

\begin{proof}[Proof of Claim \ref{claim:partialchain}]

First we prove it in the case that $M:=b-a$ divides $N$. Let region $Y$ refer to sites $[a,b-1]$, let $\rho = \Tr_{Y^c}(\ket{\eta}\bra{\eta})$, and let $\sigma = \rho \otimes\ldots \otimes \rho$ be $N/M$ copies of $\rho$ which covers all $N$ sites. Then
\begin{align}
    \Tr(H^{(j)}\sigma) =& \frac{N}{M}\bra{\eta} \sum_{i=a}^{b-2}H^{(j)}_{i,i+1}\ket\eta \nonumber \\
    &+ \sum_{k=1}^{N/M-1}\bra{\eta}H^{(j)}_{kM,kM+1}\ket{\eta} \nonumber\\
    \leq& \frac{N}{M}\bra{\eta} \sum_{i=a}^{b-2}H^{(j)}_{i,i+1}\ket\eta+\left(\frac{N}{M}-1\right),
\end{align}
where the last line follows from the fact that the interaction strength $\lVert H^{(j)}_{i,i+1} \rVert \leq 1$. Moreover, by the variational principle, $\Tr(H^{(j)}\sigma) \geq E_0^{(j)}$. These observations together yield
\begin{equation}
    \bra{\eta} \sum_{i=a}^{b-2}H^{(j)}_{i,i+1}\ket\eta \geq (M/N) (E_0^{(j)}+1)-1,
\end{equation}
which implies the statement of the claim.

Now suppose $M$ does not divide $N$. We decompose $N = s M +r$ for non-negative integers $s$ and $r < M$. Let $\sigma = \rho \otimes\ldots \otimes \rho  \otimes \ket{\nu}\bra{\nu}$ where there are $s$ copies of $\rho$ and $\ket{\nu}$ is the exact ground state of $\sum_{i=N-r+1}^{N-1} H^{(j)}_{i,i+1}$, which has energy $E_r$. Then
\begin{equation}
    \Tr(H^{(j)}\sigma) \leq s\bra{\eta} \sum_{i=a}^{i_0+M_0-2}H^{(j)}_{i,i+1}\ket{\eta}+ E_r + s.
\end{equation}
Here we invoke the variational principle twice. First, note that the expectation value of $\sum_{i=N-r+1}^{N-1} H^{(j)}_{i,i+1}$ in the state $|\phi_0^{(j)}\rangle$ (the exact ground state of the whole chain) is exactly $(r-1)E_0^{(j)}/(N-1)$. Since $\ket{\nu}$ is the exact ground state of that Hamiltonian, $E_r$ must be smaller than this quantity. Second, as before, $\Tr(H^{(j)}\sigma) \geq E^{(j)}_0$. Combining these observations yields
\begin{align}
\bra{\eta} \sum_{i=i_0}^{i_0+M-1}H^{(j)}_{i,i+1} \ket{\eta} &\geq \frac{1}{s}E_0^{(j)}\left(1-\frac{r-1}{N-1}\right)-1 \nonumber \\
&= \frac{M}{N-1}E_0^{(j)}-1.
\end{align}
\end{proof}

Now we use this claim to complete the proof of Lemma \ref{lem:Kcombined}. For any binary string $x$ we can associate a sequence of indices $1 =i_0 < i_1 < \ldots < i_{m} < i_{m+1} = N+1$ such that $x_i=x_{i_k}$ for all $k = 1,\ldots,m$ and all $i$ in the interval $[i_k, i_{k+1}-1]$. Moreover we require $x_{i_{k-1}} \neq x_{i_{k}}$. In other words, $x$ can be decomposed into substrings of consecutive 0s and consecutive 1s, with $i_j$ representing the index of the ``domain wall'' that separates a substring of 0s from a substring of 1s. The parameter $m$ is the number of domain walls. Using this notation, and letting $E_0^K = \min_j E_0^{(j)}/3$ we can rewrite
\begin{align}
    \bra{\eta_x}\sum_{i=1}^{N-1} K_{x,i}\ket{\eta_x} ={}& \sum_{k=0}^{m} \bra{\eta_x}\sum_{i=i_k}^{i_{k+1}-2} \frac{1}{3}H_{i,i+1}^{(x_{i_k})}\ket{\eta_x}\nonumber\\
    &+\sum_{k=1}^m \bra{\eta_x} I_{i_k-1,i_k}\ket{\eta_x} \nonumber\\
    \geq{}&\sum_{k=0}^{m}\left( \frac{i_{k+1}-i_k}{3(N-1)}E_0^{(x_{i_j})}-\frac{1}{3}\right)+m\nonumber\\
    \geq{}&E_0^K+\frac{2m-1}{3}.
\end{align}
For any $x$ other than $0^N$ and $1^N$, there is at least one domain wall and $m \geq 1$. Thus we can say
\begin{equation}\label{eq:gapbound}
    \bra \phi K \ket \phi \geq E_0^K + 1/3.
\end{equation}
We have shown that any state orthogonal to the states $\ket{0^N}_A \otimes |\phi^{(0)}_n\rangle$ and $\ket{1^N}_A \otimes |\phi^{(1)}_n\rangle$ will have energy at least 1/3 larger than the lowest energy state of the system. Without loss of generality, suppose $E_0^{(0)} \leq E_0^{(1)}$. Then, the ground state energy is $E_0^K = E_0^{(0)}/3$ and the ground state is $\ket{0^N}_A \otimes |\phi^{(0)}_0\rangle$ (note that in the statement of the Lemma we have $\ket{\psi^{(0)}} = \ket{\phi_0^{(0)}}$). The first excited state is either $\ket{0^N}_A \otimes |\phi^{(0)}_1\rangle$, $\ket{1^N}_A \otimes |\phi^{(1)}_0\rangle$, or lies outside the sector associated with ancillas $\ket{0^N}$ and $\ket{1^N}$, whichever has lowest energy. The three cases lead to spectral gaps of $\Delta^{(0)}/3$, $(E_0^{(1)}-E_0^{(0)})/3$, and something larger than 1/3 (due to Eq.~\eqref{eq:gapbound}), respectively. This proves all three items of the Lemma.
\end{proof}

\begin{proof}[Proof of Theorem \ref{thm:reduction}]
We begin by specifying a family of Hamiltonians, parameterized by $t \in [0,1/2]$ and defined over a chain of length $N$ with local dimension $d$.
\begin{equation}
    H^Z(t) = \sum_{i=1}^{N-1} I_i \otimes I_{i+1}-(1-t)\ket 0\bra 0_i \otimes \ket 0 \bra 0_{i+1}
\end{equation}
The ground state of $H^Z(t)$ is the trivial product state ${\ket 0}^{\otimes N}$ with ground state energy $t(N-1)$, and thus energy density $t$. The interaction strength is bounded by $1$ and the spectral gap is $1-t \geq 1/2$. 

Now we construct an algorithm for Problem \ref{prob:energydensity}. We are given $H$ as input, with associated parameters $N$ and $d$, and a lower bound $\Delta$ on the spectral gap. Let the true ground state energy for $H$ be $E$ and let $u = E/(N-1)$. We choose a value of $s$ between 0 and $1$, and we apply Lemma \ref{lem:Kcombined} to construct a Hamiltonian $K$ combining Hamiltonians $H/2$ and $H^Z(s/2)$. $K$ acts on $N$ sites, has local dimension $2d$, and has spectral gap at least $\min(\Delta, \lvert s-u \rvert (N-1),1)/6$. We are given a procedure to solve Problem \ref{prob:localprops} with $\delta = 0.9$ for a single site, i.e.~we can estimate the expectation value of any single site observable in the ground state of $K$. If $s < u$, the true reduced density matrix of $K$ will have its ancilla bits all set to 1. If $s > u$ the reduced density matrix corresponds to the reduced density matrix of $H$ with all its ancilla bits set to 0. Thus we can choose our single site operator to be the $Z_A$ operator that has eigenvalue $1$ for states whose ancilla bit is $\ket{0}$ and eigenvalue $-1$ for states whose ancilla bit is $\ket{1}$. If we have a procedure to determine $\bra{\psi}Z_A\ket{\psi}$ to precision 0.9 then we can determine the setting of one of the ancilla bits in the ground state and thus determine whether $u$ is larger or smaller than $s$. The time required to make this determination is $f(\min(\Delta, (N-1)\lvert s-u \rvert,1)/6,2d,N)$. Because we have control over $s$, we can use this procedure to binary search for the value of $u$. We assume we are given a lower bound on $\Delta$ but since we do not know $u$ \textit{a priori}, we have no lower bound on $\lvert s-u \rvert$, so we may not know how long to run the algorithm for Problem \ref{prob:localprops} in each step of the binary search. If our desired precision is $\epsilon$, we will impose a maximum runtime of $f(\min(\Delta, (N-1)\epsilon/2,1)/6,2d,N)$ for each step. Thus, if we choose a value of $s$ for which $\lvert s-u \rvert < \epsilon/2$, the output of this step of the binary search may be incorrect. After such a step, our search window will be cut in half and the correct value of $u$ will no longer be within the window. However, $u$ will still lie within $\epsilon/2$ of one edge of the window. Throughout the binary search, some element of the search window will always lie within $\epsilon/2$ of $u$, so if we run the search until the window has width $\epsilon$ and output the value $\tilde{u}$ in the center of the search window, we are guaranteed that $\lvert u - \tilde{u} \rvert \leq \epsilon$. The number of steps required is $O(\log(1/\epsilon))$ and the time for each step is $f(\min(2\Delta, (N-1)\epsilon,2)/12,2d,N)$, yielding the statement of the theorem. 
\end{proof}


\section{Discussion}

Our results paint an interesting landscape of the complexity of approximating ground states of gapped nearest-neighbor 1D Hamiltonians locally. On the one hand, we show all $k$-local properties of the ground state can be captured by an MPS with only constant bond dimension, an improvement over the $\text{poly}(N)$ bond dimension required to represent the global approximation. This constant scales like a polynomial in $k$ and $1/\epsilon$, when parameters like $\Delta$, $\xi$, and $d$ are taken as constants. On the other hand, we give evidence that, at least for the case where the Hamiltonian is translationally invariant, finding the local approximation may not offer a significant speedup over finding the global approximation: we have shown that the ability to find even a constant-precision estimate of local properties would allow one to learn a constant-precision estimate of the ground state energy with only $O(\log(N))$ overhead. This reduction does not allow one to learn any global information about the state besides the ground state energy, so it falls short of giving a concrete relationship between the complexity of the global and local approximations. Nonetheless, the reduction has concrete consequences. In particular, at least one of the following must be true about translationally invariant gapped Hamiltonians on chains of length $N$:
\begin{enumerate}[(1)]
    \item The ground state energy can be estimated to $O(1)$ precision in $O(\log(N))$ time.
    \item Local properties of the ground state cannot be estimated to $O(1)$ precision in time independent of $N$. 
\end{enumerate}
In particular, the second item, if true, would seem to imply that, in the translationally invariant case when $N \rightarrow \infty$, local properties cannot be estimated at all.

Indeed, it is when the chain is very long, or when we are considering the thermodynamic limit directly that our results are most relevant. In the translationally invariant case as $N\rightarrow \infty$, our first proof method (Theorem \ref{thm:improvedbd}) yields a local approximation that is a translationally invariant MPS. However, the MPS is non-injective and the state is a macroscopic superposition on the infinite chain. Thus the bulk tensors alone do not uniquely define the state and specification of a boundary tensor at infinity is also required \cite{vanderstraeten2019tangent,zauner2018variational}. Our second proof method (Theorem \ref{thm:mainthm}), on the other hand, yields a periodic MPS (with period $O(\log(k/\epsilon^2))$) that is injective and can be constructed by a constant-depth quantum circuit made from spatially local gates. If we allow the locality of the gates to be $O(\log(k/\epsilon^2))$, then the circuit can have depth 2, as in Figure \ref{fig:constantdepthcircuit}. If we require the locality of the gates be only a constant, say 2, then the circuit can have depth $\text{poly}(k,1/\epsilon)$.

The fact that the approximation is injective perhaps makes the latter method more powerful. Injective MPS are the exact ground states of some local gapped Hamiltonian \cite{fannes1992finitely, perez2006matrix}. Additionally, non-injective MPS form a set of measure zero among the entire MPS manifold, so variational algorithms that explore the whole manifold are most compatible with an injective approximation. In fact, since the approximation can be generated from a constant-depth circuit, the result justifies a more restricted variational ansatz using states of that form. This ansatz could provide several advantages over MPS in terms of number of parameters needed and ability to quickly calculate local observables, like the energy density. 
However, algorithms that perform variational optimization of the energy density generally suffer from two issues, regardless of the ansatz they use. First, they do not guarantee convergence to the global minimum within the ansatz set, and second, even when they do find the global minimum, the output does not necessarily correspond to a good local approximation. This stems from the fact that a state that is $\epsilon$-close to the ground state energy density may actually be far, even orthogonal, to the actual ground state. Therefore, even a brute-force optimization over the ansatz set cannot be guaranteed to give any information about the ground state, other than its energy density. 

This leaves open many questions regarding the algorithmic complexity of gapped local 1D Hamiltonians. For the general case on a finite chain, can one find a local approximation to the ground state faster than the global approximation? For translationally invariant chains, can one learn the ground state energy to $O(1)$ precision in $O(\log(N))$ time, and can one learn local properties in time independent of the chain length? Relatedly, in the thermodynamic limit, can one learn an $\epsilon$-approximation to the ground state energy density in $O(\log(1/\epsilon))$ time, and can one learn local properties at all? These are interesting questions to consider in future work. 

We would like to conclude by drawing the reader's attention to independent work studying the same problem by Huang \cite{huang2019approximating}, which appeared simultaneously with our own.

\begin{acknowledgments}

We thank Thomas Vidick for useful discussions about this work and its algorithmic implications. AMD gratefully acknowledges support from the Dominic Orr Fellowship and the National Science Foundation Graduate Research Fellowship Program under Grant No. DGE‐1745301. This work was supported by NSF and Samsung. The Institute for Quantum Information and Matter (IQIM) is an NSF Physics Frontiers Center.
\end{acknowledgments}

\bibliographystyle{abbrvnat}
\bibliography{references}
\end{document}